\pdfoutput=1
\documentclass[aip,cha,preprint]{revtex4-1}
\usepackage{fullpage,amsthm,amsmath,amsfonts,bm,times,color,bbm}
\usepackage{graphicx,epstopdf,subfigure}

\makeatletter
  \theoremstyle{plain}
  \newtheorem{prop*}{\protect\propositionname}

\makeatother
  \providecommand{\propositionname}{Proposition}

\newcommand{\avg}[1]{\left\langle #1\right\rangle}

\newcommand{\T}[1]{\textrm{#1}}
\newcommand{\jj}{\mathbbm{j}}

\newcommand{\red}[1]{{\color{black} #1}} % makes red change to black

\begin{document}

%\title{Patterns of patterns of synchronization: noise induced switching behavior (of) phase multistability in rings of reactively coupled oscillators}
\title{Patterns of patterns of synchronization: Noise induced \red{attractor} switching in rings of coupled nonlinear oscillators}

%\author{JE, AC, MP, MM, JC, RD, ...?}
\author{Jeffrey Emenheiser}
\affiliation{Complexity Sciences Center, University of California, Davis, CA USA 95616}
\affiliation{Department of Physics, University of California, Davis, CA USA 95616}

\author{Airlie Chapman}
\affiliation{William E. Boeing Department of Aeronautics \& Astronautics, University of Washington, Seattle, WA USA 98195}

\author{M{\'a}rton P{\'o}sfai}
\affiliation{Complexity Sciences Center, University of California, Davis, CA USA 95616}
\affiliation{Department of Computer Science, University of California, Davis, CA USA 95616}

\author{\\James P. Crutchfield}
\affiliation{Complexity Sciences Center, University of California, Davis, CA USA 95616}
\affiliation{Department of Physics, University of California, Davis, CA USA 95616}
\affiliation{Santa Fe Institute, Santa Fe, New Mexico, USA 87501} % 1399 Hyde Park Road,

\author{Mehran Mesbahi}
\affiliation{William E. Boeing Department of Aeronautics \& Astronautics, University of Washington, Seattle, WA USA 98195}

\author{Raissa M. D'Souza}
\affiliation{Complexity Sciences Center, University of California, Davis, CA USA 95616}
\affiliation{Department of Computer Science, University of California, Davis, CA USA 95616}
\affiliation{Santa Fe Institute, Santa Fe, New Mexico, USA 87501} % 1399 Hyde Park Road,
\affiliation{Department of Mechanical and Aerospace Engineering, University of California, Davis, CA USA 95616}

\date{\today}

\begin{abstract}
\red{Following the long-lived qualitative-dynamics tradition of explaining behavior in complex systems via the architecture of their attractors and basins, we investigate the patterns of switching between qualitatively distinct trajectories in a network of synchronized oscillators. Our system, consisting of nonlinear amplitude-phase oscillators arranged in a ring topology with reactive nearest neighbor coupling, is simple and connects directly to experimental realizations.} We seek to understand how the multiple stable synchronized states connect to each other in state space by applying Gaussian white noise to each of the oscillators' phases.
To do this, we first identify a set of locally stable limit cycles at any given coupling strength. For each of these attracting states, we analyze the effect of weak noise via the covariance matrix of deviations around those attractors.  We then explore the noise-induced attractor switching behavior via numerical investigations. For a ring of three oscillators we find that an attractor-switching event is always accompanied by the crossing of two adjacent oscillators' phases.  For larger numbers of oscillators we find that the distribution of times required to stochastically leave a given state falls off exponentially, and we build an attractor switching network out of the destination states as a coarse-grained description of the high-dimensional attractor-basin architecture.
\end{abstract}


\maketitle

\begin{quotation}
Proper functioning of large-scale complex systems, from metabolism to global economics, relies on the coordination of interdependent systems. Such coordination -- the emergence of synchronization in coupled systems -- is itself an important and widely studied collective behavior. However, predicting system behavior and controlling it to maintain function or mitigate failure present severe challenges to contemporary science. Prediction and control depend most directly on knowing the architecture of the stable and unstable behaviors of such high-dimensional dynamical systems. To make progress, here we explore limit-cycle attractors arising when \red{ring} networks of nonlinear oscillators synchronize, demonstrating how synchronization emerges and stabilizes and laying out the combinatorial diversity of possible synchronized states. We capture the global attractor-basin architecture of how the distinct synchronized states can be reached from each other via attractor switching networks.
\end{quotation}

\section{Introduction}\label{sec:introduction}

From the gene regulatory networks that control organismal development~\cite{davidson2005gene} and the coherent oscillations between brain regions responsible for cognition~\cite{gray1992synchronization} to the connected technologies that support critical infrastructure~\cite{rinaldi2001identifying,rosato2008modelling,gonzalez2015interdependent}, systems at many scales of modern society rely on the coordination of the dynamics of interdependent systems. Analyzing the mechanisms driving such complex networks presents serious challenges to dynamical systems, statistical mechanics, and control theory, including but not limited to the overtly high dimension of their state spaces. This precludes directly identifying and visualizing their attractors and attractor-basin organization. Moreover, without knowledge of the latter large-scale architecture, predicting network behavior, let alone developing control strategies to maintain function or mitigate failure, is impossible.

To shed light on these challenges, we explore limit-cycle attractors arising when \red{rings of coupled nonlinear oscillators} synchronize. We demonstrate how synchronization emerges and stabilizes and identify the diversity of synchronized states. We probe the global attractor-basin architecture by driving the networks with noise, capturing how the distinct synchronized states can be reached from each other via what we call attractor switching networks. The analysis relies on the use of limit-cycle attractors to define coarse-grained units of system state space.

\red{ In this way, our study of attractor-basin architecture falls in line with the methods of qualitative dynamics introduced by Poincare~\cite{jpc8}. Confronted with unsolvable nonlinear dynamics in the three-body problem, Poincare showed that system behaviors are guided and constrained by invariant state space structures -- fixed point, limit cycle, and chaotic attractors -- and their arrangement in state space -- basins of attraction and their separatrices. The power of his qualitative approach came in providing a concise description of all possible behaviors of a system, without requiring detailed system solutions. His architectural approach is more recently expressed in terms of Smale basic sets~\cite{jpc6,jpc7} and Conley's index theory~\cite{jpc4,jpc5}. These show that any system decomposes into recurrent and wandering subspaces in which the behavior is a gradient flow. In short, there is a kind of Lyapunov function over the entire state space, underlying the architectural view of attractors and their basins. This view is so basic to our modern understanding of nonlinear complex systems that it has been rechristened as the ``Fundamental Theorem of Dynamical Systems''~\cite{jpc3}. As we will see, our analytical study of oscillator arrays appeals to Lyapunov functions to locally analyze limit cycle stability and noise robustness, while our numerical explorations allow us to knit together the stable attractors into a network of stable oscillations, connected by particular pathways that facilitate switching between them.

Practically, complicated attractor-connectivity architectures can be probed via external controls or added noise. We focus on the latter here, following recent successful explorations of noise-driven large-scale systems. For example, the analysis of bistable gene transcription networks showed that attractor switches can be induced by periodic pulses of noise~\cite{jpc1}. Another recent study of networks of pulse-coupled oscillators showed that unstable attractors become prevalent with increasing network size and the attractors are closely approached by basin tendrils of other attractors. Thus, arbitrarily small noise can lead to switching between attractors~\cite{jpc2}. Our explorations illustrate the theoretical foundations and complements the newer works by focusing on the dynamics of synchronization.}

Synchronization between oscillators is itself an important and widely studied collective behavior of coupled systems~\cite{PIK03}, with examples ranging from neural networks~\cite{HOP97} to power grids~\cite{MOT13}, clapping audiences~\cite{NED00}, and fireflies flashing in unison~\cite{MIR90}. Although different in scope and nature, all of these examples can be modeled as \red{coupled oscillators}. Decades of research has revealed that a system of coupled oscillators may produce a rich variety of behaviors; in addition to full synchronization, more complex patterns may emerge, including chaos~\cite{MAT91}, chimera states~\cite{ABR04,HAG12}, and cluster synchronization~\cite{ARE06,WAN09}. Here, we study rings of oscillators -- a system that exhibits multiple stable synchronized patterns called rotating waves~\cite{ERM85}. Rings of oscillators have been extensively studied~\cite{BRE97,ROG04,WIL06,CHO10,HA12,SHA15}; our contribution in this respect focuses on \emph{reactively} coupled amplitude-phase oscillators and the organization of their attractors, basins, and noise-driven basin transitions.

\red{Reactive coupling, in the context of electromechanical oscillators, is that which does not dissipate energy, such as ideal elastic and electrostatic interactions between devices\cite{LIF08}. A primary motivation of this work is to connect with experiment, using reactive coupling to characterize systems of nearest-neighbor coupled rings of nanoscale piezoelectric oscillators~\cite{CRO04}. Recent experiments investigated synchronous behavior of two such nanoelectromechanical systems (NEMS)~\cite{MAT14}, and it is expected that in the near future larger rings and more complex arrangements will be realized.~\cite{fonUpcoming} In the context of the complex values used to model these oscillators, reactive coupling means that the coefficient of the Laplacian coupling terms is purely imaginary. This coupling is captured, between Landau-Stuart oscillators, as a special case of the complex Ginzburg-Landau equation, which describes a wide range of physical phenomena~\cite{kuramoto1984chemical, ARA02}.}

If no noise is present, the system settles at one of its stable steady states. Exactly which stable state depends on initial conditions. In the presence of noise, the long-term behavior of the system is no longer characterized by deterministic attractors. Depending on the level of noise three possible scenarios may emerge: (i)~if the noise is small, the system fluctuates around an attractor; (ii)~if the noise is strong, the system is randomly pushed around in the state space suppressing the deterministic dynamics; and (iii)~intermediate levels of noise cause the system to fluctuate around an attractor and occasionally jump to the basin of attraction of a different attractor. The latter scenario \red{suggests} a coarse-grained description of the system’s global organization: we specify the effective ``macrostate'' of the system by the attractor it fluctuates around, and we map out the likelihood of transitions to other attractors. These transitions form an attractor switching network (ASN) capturing the coarse-grained dynamics of the system. Noise and external perturbation induced jumps in the ASN have been suggested as a feasible strategy to control large-scale nonlinear systems~\cite{COR13,PIS14,WEL15}.

Our goal is to study the fluctuations of the system and attractor switching  in the presence of additive uncorrelated Gaussian noise in the phases of oscillators. Setting up the analysis, we introduce the system in Sec.~\ref{Sec:Deterministic}, finding the available patterns of synchronization in Sec.~\ref{Sec:Solutions} and their local stability in Sec.~\ref{Sec:Stability}. We introduce noisy dynamics in Sec.~\ref{Sec:Stochastic}. Based on the linearized dynamics we derive a closed-form expression that predicts the system’s response to small noise in Sec.~\ref{Sec:Covariance}. We demonstrate that the attractor switching occurs via a phase-crossing mechanism Sec.~\ref{Sec:Phase Crossing}. This motivates the coarse-graining of state space such that we can finally compile an ASN for a network of $N=11$ oscillators in Sec. \ref{Sec:The ASN}.

\section{Deterministic Dynamics}
\label{Sec:Deterministic}

We study rings of reactively coupled oscillators that are governed by
\red{
\begin{align}
\label{eq:envelope}
\frac{dA_i}{dt}= -\frac{1}{2}A_i &+ \jj\alpha\vert A_i\vert^2A_i + \frac{A_i}{2|A_i|}  + \frac{\jj\beta}{2}\left[A_{i+1}-2A_{i} + A_{i-1}\right],
\end{align}
}
where $A_i\in\mathbb C$ describes the amplitude and phase of the $i^\text{th}$ oscillator \((i = 1,2,...,N)\) \red{and \(\jj=\sqrt{-1}\)}. The first three terms describe the oscillators in isolation: the first is the linear restoring force, the second term is the first nonlinear correction known as the Duffing nonlinearity, and the third term is a saturated feedback that allows the system to sustain oscillatory motion. The fourth term expresses the inter-oscillator \red{feedback}: the oscillators are diffusively coupled to their nearest neighbors with purely imaginary coefficient. Equation~(\ref{eq:envelope}) was derived to describe the slow modulation of rapid oscillations of a system of NEMS -- sometimes referred to as an envelope or modulational equation~\cite{LIF08}.

Although Eq. (\ref{eq:envelope}) presents a compact representation of the dynamics, it is in this case more insightful, and useful, to isolate the amplitude and phase components of the representative complex state. We therefore separate the dynamics of the amplitudes \(a_i\) in vector \(\bm{a}\in\mathbb{R}^N\) and those of the phases \(\phi_i\) in vector \(\bm{\phi}\in\mathbb{R}^N\) according to \red{\(A_i = a_ie^{\jj\phi_i}\)}. The system then evolves according to
\begin{align}
\label{eq:amplitude}
\frac{da_i}{dt} &= -\frac{a_i-1}{2} - \frac{\beta}{2}\bigg[a_{i+1}\sin\left(\phi_{i+1}-\phi_i\right)  + a_{i-1}\sin\left(\phi_{i-1}-\phi_i\right)\bigg], \\
\label{eq:phase}
\frac{d\phi_i}{dt} &= \alpha a_i^2 + \frac{\beta}{2}\bigg[\frac{a_{i+1}}{a_i}\cos\left(\phi_{i+1}-\phi_i\right)  + \frac{a_{i-1}}{a_i}\cos\left(\phi_{i-1}-\phi_i\right)-2\bigg].
\end{align}

These equations make it clear that in the absence of coupling (\(\beta = 0\)), each amplitude \(a_i\) will settle to unity, and all phases oscillate with constant frequency \(\alpha\). This frequency, proportional to the square of the oscillator's amplitude, comes from the device's nonlinear restoring force -- the Duffing nonlinearity. This effect is accordingly referred to as nonlinear frequency pulling. We now proceed to find solutions of the dynamics with nonzero coupling.

\subsection{Analytic Solutions: Rotating Waves}
\label{Sec:Solutions}

\begin{figure*}[t]
	\centering
	\includegraphics[width=\textwidth]{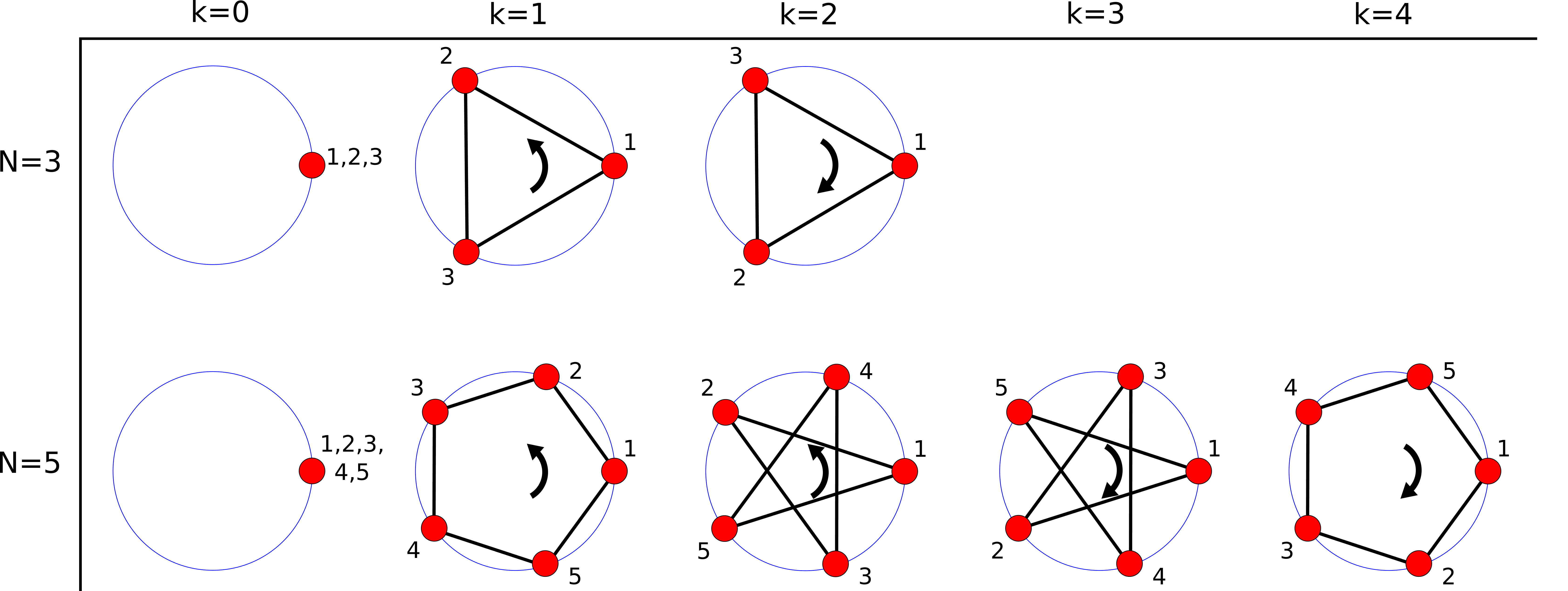}
	\caption{{\bf Rotating wave solutions.} Relative phase diagrams representing the N rotating wave solutions for systems of N = 3 and N = 5 oscillators. The blue circles represent the unit circle in the complex plane, and each red circle represents a value of an oscillator envelope \red{\(A_i = a_ie^{\jj\phi_i}\)} \red{labeled by its index}. The black lines indicate coupling between neighboring oscillators on the ring network. Since the oscillator positions are repeated, the black arrows indicate in which direction the oscillators are numbered. All rings are represented as if the first oscillator has zero phase, on the far right of the unit circle. The \(k=0\) diagrams have only one red circle and no black lines because all oscillators have the same phase and are thus all represented by the single red circle.}	
	\label{fig:wavediagrams-k}
\end{figure*}

%Here we show that the ring of \(N\) reactively coupled oscillators exhibits multiple unit amplitude rotating wave solutions characterized by integer wavenumbers \(k\in[0,N-1]\).

To view self-organized patterns of synchronization of these nonlinear oscillators, we consider only the weak coupling regime, with positive nonlinear frequency pulling: $|\beta| \le \alpha \sim 1$. This selection is heavily motivated by upcoming experimental realizations of the system~\cite{fonUpcoming} and ensures that the internal nodal dynamics are not dominated by coupling terms. With zero coupling (\(\beta=0\)), each oscillator will follow its own limit cycle, and the composite attractor will have \(N\) dimensions -- one corresponding to the phase of each oscillator. For small but nonzero coupling (\(\beta\rightarrow0\)) we expect the leading order effect to be in the dynamics of phases. As these are limit cycles, displacements along phase are not restored except through the coupling edges. \red{Solving Eqs. (\ref{eq:amplitude}) and (\ref{eq:phase}) for sets of stationary phase differences with fixed unit amplitudes,}
\begin{align}
\frac{da_i}{dt}\bigg|_{\mathbf{a}=\mathbf{1}} = 0 &= -\frac{\beta}{2}\Big[\sin\Delta_{i} - \sin\Delta_{i-1}\Big],\\
\left(\frac{d\phi_{i+1}}{dt} - \frac{d\phi_i}{dt}\right)\Big|_{\mathbf{a}=\mathbf{1}} = 0 &= \frac{\beta}{2}\Big[\cos\Delta_{i+1} - \cos\Delta_{i-1}\Big],
\end{align}
\red{where \(\Delta_{i} \equiv \phi_{i+1}-\phi_i\) is the (signed) phase difference between adjacent oscillators \(i\) and \(i+1\).} These conditions are satisfied if and only if every other phase difference is equal to some \(\Delta\), where the other phase differences are either \(\pi-\Delta\) or also \(\Delta\). \red{Note that these conditions are independent of \(\beta\), so the solutions will be valid for all coupling strengths.} Because the ring is a periodic lattice and the sum of all \(N\) phase differences must be an integer multiple of \(2\pi\), limit cycles that satisfy the \(\pi-\Delta\) condition for alternating phase differences may exist if and only if the number of nodes is an integer multiple of four. To ease comparison of attractors in systems with various numbers of nodes, we limit our subsequent discussion to the solutions defined wholly by a single phase difference \(\Delta\) supported across all edges, implying that \(N\) may not be a multiple of four.

For limit cycles where all phase differences are identical, {\it i.e.,} \(\Delta_{i} = \Delta\text{ for all } i\), the periodic boundary condition requires \(\Delta\) to be an integer multiple of \(2\pi/N\), giving precisely \(N\) unique states of this sort. These states follow the trajectory
\begin{align}
\label{eq:ampsolution}
a_i(t) &= 1,\\
\label{eq:phasesolution}
\phi_i(t) &= \phi_i(0) + \left(\alpha +\beta\left(\cos\frac{2k\pi}{N}-1\right)\right)t,
\end{align}
specific to a particular wavenumber \(k\). These are the expected rotating wave solutions. Each rotating wave has a fixed phase configuration, with phase differences of \(2\pi k/N\), represented in Eq. (\ref{eq:phasesolution}) as initial phases \(\phi_i(0)\). The form of reactive coupling causes the frequency of oscillation also to be dependent upon the wavenumber. Noting that the phase difference \(\Delta\) is invariant under \(k\rightarrow N+k\), we choose to make the restriction \(0\le k <N\).

Relative phase diagrams representing the \(N\) unique configurations for \(N=3\) and \(N=5\) oscillator rings are shown in Fig. 1. In these, each oscillator is represented as a point on the unit circle in the complex plane, with edges connecting adjacent, coupled oscillators. Each edge connects oscillators with an arc length separation equal to the phase difference \(\Delta = 2\pi k/N\). We see that, for instance, \(N=5\) and \(k = 2\) or \(3\) results in next nearest neighbors being closer in phase than nearest neighbors. This is a general result; as \(k/N\rightarrow1/2\), neighboring oscillators will have a phase difference of \(\pi\) and next nearest oscillators have nearly equivalent phases. This is locally out-of-phase sychronization, in contrast to \(k=0\), which is completely in-phase synchronization, {\it i.e.,} zero phase difference between neighboring oscillators. We also see a symmetry in wave numbers \(k\) and \(N-k\). These waves travel in opposite directions around the ring; the phase configurations amount to a relabeling of oscillators, represented in the Fig. \ref{fig:wavediagrams-k} by arrows indicating the direction of labeling. Just as the wavenumber represents the number of wavelengths of the rotating wave along the length of the ring, it may be interpreted as the winding number of the ring about the origin when represented in the complex plane as in Fig. \ref{fig:wavediagrams-k}.

We have thus discovered \(N\) synchronized states that are possible nodes of the global attractor switching network and which the system might visit once noise is included in the dynamics. Although motivated by the weak coupling limit, these rotating waves are valid solutions at all coupling magnitudes. Note that there can be solutions that do not converge to the unit amplitude states enumerated here. With sufficiently weak coupling, however, oscillator amplitudes in attractors are in fact confined to stay within a distance of order \(\beta\) from unity. This is shown in Appendix~\ref{App:AttractorAmplitudes} using a Lyapunov-like potential function. Having enumerated such candidate synchronized limit cycles, we need to determine their stability in order to identify those that we expect the noisy system to visit for extended times.

\subsection{Local Stability: Attracting Patterns}
\label{Sec:Stability}

\begin{figure}[t]
	\centering
	\includegraphics[width=0.7\columnwidth]{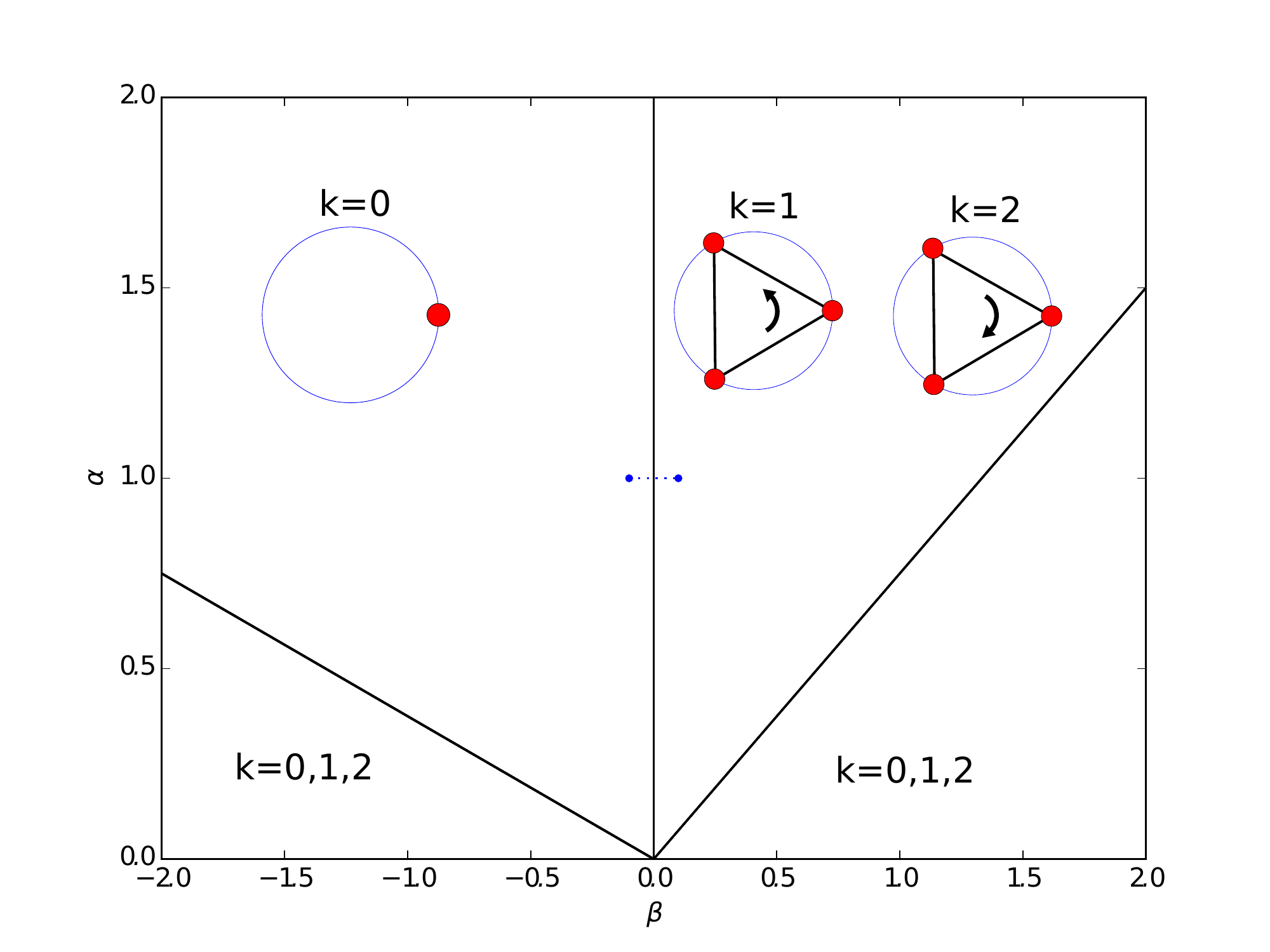}
	\caption{{\bf Regions of stability.} The stable rotating waves on the \(N=3\) ring in each of four regions of nonlinearity/coupling \((\alpha, \beta)\) parameter space, separated by solid black lines. The \(k=0\) wave is locally stable for all \(\beta<0\) and for \(\beta > \alpha \csc^2(\pi/3)\). The \(k=1,2\) waves are locally stable for all \(\beta>0\) and for \(\beta<-\alpha\csc^2(\pi/3)\sec(\pi/3)\). The blue dotted line indicates the parameters of likely experimental realizations. Our simulations were done at the endpoints \(\alpha = 1,\,\beta = \pm0.1\).}
	\label{fig:regions}
\end{figure}

Here we show that the stability of each rotating wave/pattern of synchronization to small perturbations is equivalent to finding the sign of \(\beta\cos(2\pi k/N)\). We then characterize the linear response of these waves to uncorrelated, white Gaussian noise on the oscillator phases and find that the \(k = 0\) and \(k\red{\approx} N/2\) waves amplify noise least in their respective stable regimes.

Linearizing \red{Eqs. (\ref{eq:amplitude}) and (\ref{eq:phase})} around any point on the limit cycle defined by wavenumber \(k\), we find the \(2N\times2N\) matrix \(F\) that governs the linear evolution of small deviations from that limit cycle. We write this matrix in block form, such that \(F_{ij}\) is the \(2\times 2\) matrix corresponding to the dependence of deviations in oscillator \(i\) on deviations in oscillator \(j\).
\begin{multline}
\label{eq:Linearized Dynamics without Noise}
\frac{d}{dt}\left[\begin{array}{c} \delta a_i\\\delta\phi_i\end{array}\right] = \sum_jF_{ij}\left[\begin{array}{c} \delta a_j\\\delta\phi_j\end{array}\right]
\\
=  \frac{1}{2}\sum_{j=0}^N\left[\begin{array}{cc} -I_{ij} - M_{ij}\beta\sin(2\pi k/N) & L_{ij}\beta\cos(2\pi k/N) \\ 4\alpha I_{ij} - L_{ij}\beta\cos(2\pi k/N) &-M_{ij}\beta\sin(2\pi k/N)\end{array}\right]\left[\begin{array}{c} \delta a_j\\\delta\phi_j\end{array}\right],
\end{multline}
where \(I\) is the \(N\times N\) identity matrix, \(L\) is the \(N\times N\) unweighted ring Laplacian matrix, and \(M\) is an \(N\times N\) next-nearest-neighbor oriented incidence matrix of the ring,
\begin{equation}
L_{ij} = \begin{cases} 2 & i=j \\ -1 & i=j\pm 1\\0&\text{otherwise}\end{cases}\qquad\,M_{ij} = \begin{cases} 1 & i=j+1\\-1 & i=j-1 \\ 0 & \text{otherwise .}\end{cases}
\end{equation}

 The local stability of each rotating wave is then determined by the signs of the eigenvalues of \(F\). While this is straightforward to do numerically, we find that exclusion of terms varying with \(M\) will not effect any changes of sign, as detailed in Appendix \ref{App:Linearization}. We denote this simplified matrix \(\tilde{F}\) and transform \(\tilde{F}\) by a matrix \(U\) to diagonalize the Laplacian \(L\), leaving a \(2\times2\) linear dynamics for each Laplacian mode.  The matrices \(L\) and \(M\) are not mutually diagonalizable, so this cannot be done with the full linearization \(F\). Deviations in these Laplacian modes are then governed by
\begin{equation}
\left(U\tilde{F}U^{-1}\right)_{ii} = \frac{1}{2}\left[\begin{array}{cc} -1 & \rho_i\beta\cos(2\pi k/N)\\ 4\alpha - \rho_i\beta\cos(2\pi k/N) &0\end{array}\right].
\end{equation}
where \(\rho_i = 4\sin^2\left(\frac{\lfloor i/2\rfloor\pi}{N}\right)\) are the eigenvalues of \(L\) \red{for the ring coupling topology} (and \(\lfloor\cdot\rfloor\) is the floor operation).

Defining \(x_i = \rho_i\beta\cos(2\pi k/N)\), we see that \(\tilde{F}\) represents stable trajectories if and only if all its eigenvalues \(\mu_{\pm,i} = -\frac{1}{4}\left(1\pm\sqrt{16\alpha x_i-4x_i^2+1}\right)\) have negative real part. That is, the rotating wave is stable if and only if \(4\alpha x_i-x_i^2<0\) for all Laplacian modes. The mode associated with \(\rho_1=0\), giving \(\mu_{-,1} = 0\), may in fact be ignored. This zero eigenvalue corresponds to the freedom of deviations along the limit cycle and is explicitly removed in Appendix \ref{App:Linearization} by stabilizing this allowed nullspace of \(\tilde{F}\). Then, there are two regimes in which a mode of the modified dynamics is stable: \(x_i<0\) and \(x_i>4\alpha\).

Now, we see that all \(\rho_{i>1}\) are strictly positive and therefore all \(x_{i>1}\) will be of the same sign as \(\beta\cos(2\pi k/N)\). With a given sign of the coupling \(\beta\), all wavenumbers \(k\) satisfying \(\beta\cos(2\pi k/N)<0\) will correspond to stable rotating waves for all coupling magnitudes \(|\beta|\).

A wave solution is also stable if \(\beta\cos(2k\pi/N)\) is large enough such that the smallest nonzero Laplacian eigenvalue \(\rho_2\) corresponds to \(x_2>4\alpha\). This occurs when \(\beta\cos(2k\pi/N)>\alpha\csc^2(\pi/N)\) (and requires \(\beta\cos(2\pi k/N)>0\)). This scenario clearly corresponds to large coupling magnitudes, which we are not considering here.

Figure \ref{fig:regions} portrays the above stability conditions in \((\alpha,\beta)\) parameter space of the \(N=3\) oscillator ring. There are four distinct regions: large or small coupling-to-nonlinearity ratio, with positive or negative coupling. The more nearly in- (out-of-) phase adjacent node oscillations are stable with small, negative (positive) coupling and become stable with positive (negative) coupling at some critical coupling magnitude proportional to the nonlinear coefficient \(\alpha\). Each rotating wave has a critical ratio \(|\beta|/\alpha\), proportional to \(\sin^2(\pi/N)\), above which the wave is stable for either sign of \(\beta\). As such, the required coupling magnitudes for this regime increase with \(N\). These boundaries were found analytically as described above and corroborated by diagonalizing the original linearization \(F\) numerically, validating the process of studying the modified dynamics in
\(\tilde{F}\).

\section{Stochastic Dynamics}
\label{Sec:Stochastic}

We have so far found attractors of the deterministic dynamics of rings of reactively coupled nonlinear oscillators. This identifies \red{orbits that may have some global importance in the system}, but gives little indication of the higher-level state space architecture. We investigate this organization by applying noise to the oscillators’ phases, first weakly and then strongly enough to induce distinct jumps between attracting limit cycles.

Specifically, we focus on the analysis of one of the most
ubiquitous and well-modelled forms of disturbances, namely white Gaussian
noise. The injection point is assumed to be an additive time-varying signal on the phases. This generates a perturbed dynamics of the form
%\begin{eqnarray*}
%\frac{da}{dT} & = & f_{a}(a,\phi),\\
%\frac{d\phi}{dT} & = & f_{\phi}(a,\phi)+w,
%\end{eqnarray*}

\begin{align}
\label{eq:amplitudeWithNoise}
\frac{da_i}{dt} &= -\frac{a_i-1}{2} - \frac{\beta}{2}\bigg[a_{i+1}\sin\left(\phi_{i+1}-\phi_i\right) + a_{i-1}\sin\left(\phi_{i-1}-\phi_i\right)\bigg],\\
\label{eq:phaseWithNoise}
\frac{d\phi_i}{dt} &= \alpha a_i^2 + \frac{\beta}{2}\bigg[\frac{a_{i+1}}{a_i}\cos\left(\phi_{i+1}-\phi_i\right) + \frac{a_{i-1}}{a_i}\cos\left(\phi_{i-1}-\phi_i\right)-2\bigg] + w_i.
\end{align}
where \(w_i(t)\) is an element of $\bm{w}(t)\in\mathbb{R}^{N}$: an uncorrelated zero mean i.i.d.
Gaussian random process with covariance matrix $\sigma^2 I_N$.

\subsection{Weak Noise Response}
\label{Sec:Covariance}

\begin{figure*}[t]
	\centering
	\includegraphics[width=\textwidth]{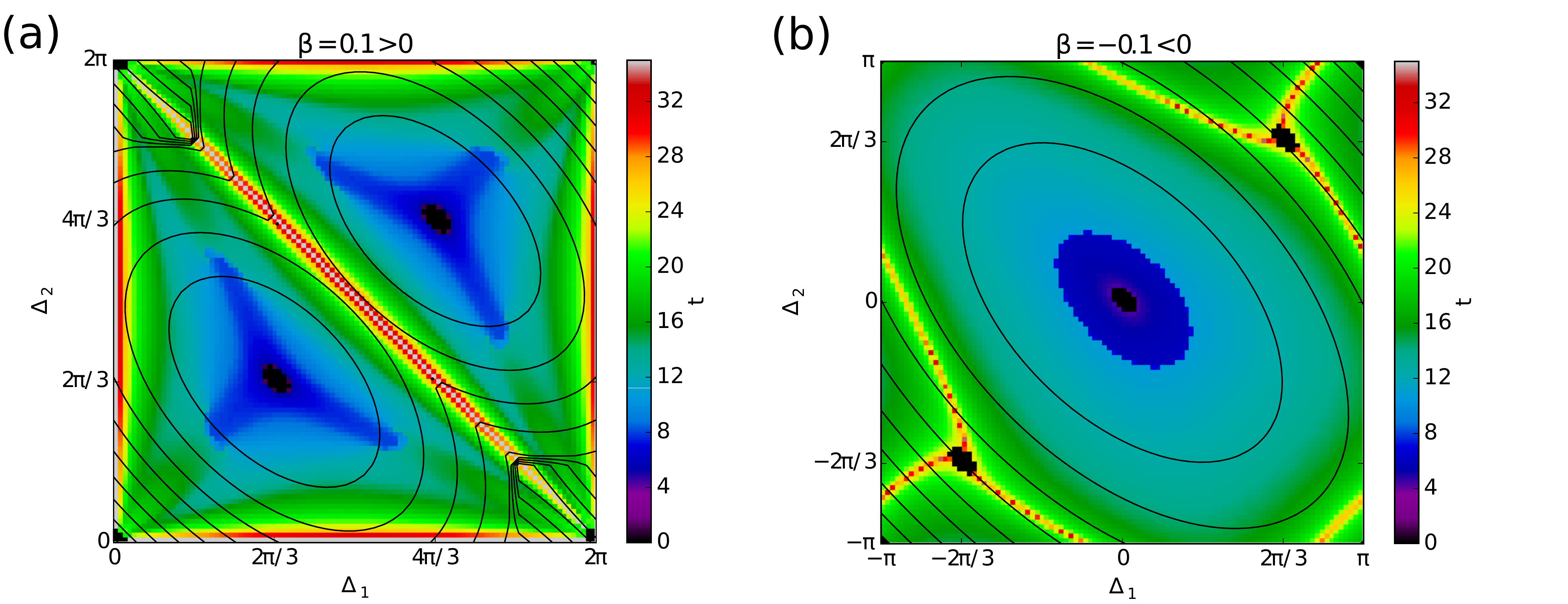}
	\caption{{\bf Potential function and times to convergence.} Level sets of the potential defined in Eq. (\ref{eq:Potential}) for \(N=3\) nodes, across the two free phase differences, \(\Delta_1, \Delta_2\), with unit oscillator amplitudes. Shown are results for \(\alpha=1\) with (a) \(\beta=0.1\), where there are two stable states, and (b) \(\beta=-0.1\), with one stable state. As the potential functions are not guaranteed to be good indicators of the convergence time, results from direct implementation of the deterministic dynamics are also shown. Here the colormap background indicates the time for the  deterministic system to converge to near a stable rotating wave.}
	\label{fig:TimeToConvergence}
\end{figure*}

Having identified attractors -- the stable rotating waves, we begin to study the basin architecture by characterizing the system's response to weak noise at each of the stable rotating waves, finding that \(k=0\) has the least amplification of noise when \(\beta<0\) and \(k\approx N/2\) has the least for \(\beta>0\).

Close to an attractor, the dynamics can be predominately described
by its linearization:
\begin{equation}
\frac{d}{dt}\left[\begin{array}{c}
\delta a\\
\delta\phi
\end{array}\right]=F\;\left[\begin{array}{c}
\delta a\\
\delta\phi
\end{array}\right]+\left[\begin{array}{c}
0\\
I
\end{array}\right]w,\label{eq:Linearized Dynamics with Noise}
\end{equation}
where $F$ is the linearized state matrix of Eq.~(\ref{eq:Linearized Dynamics without Noise}) associated with the wavenumber \(k\)
and weak coupling $\beta$, and the additive term describes injection
of noise into the phase dynamics. The local amplification of the noise can be described by the steady state covariance. Specifically, the expectation of the outer product of deviations from the attracting limit cycles is
\begin{equation}
P = \lim_{t\rightarrow\infty}\mathbb{E}\left\{
\left[\begin{array}{c}\delta a\\\delta\phi\end{array}\right]
\left[\begin{array}{c}\delta a\\\delta\phi\end{array}\right]^T\,
\right\}.
\end{equation}
Small entries in $P$ indicate a good robustness of the attractor
to noise as the steady state variances and cross-covariances of the
dynamics are small, representing small deviations around the equilibrium.

The eigenvalues of $P$ represent the axis lengths of the covariance ellipse. \red{Large eigenvalues are associated with directions of large noise amplification when compared to eigenvalues which are close to zero.}

For distinct pairs of rotating waves, \(k\) and \(N-k\), the covariance and associated eigenvalues are the same, exhibiting a common robustness to noise. This underlying symmetry indicates that equal time will be spent between attractors
$\Delta_{k}$ and $\Delta_{N-k}$ when driven by basin switching noise, to be discussed in the following sections.

As observed in the deterministic linearization, Eq. (\ref{eq:Linearized Dynamics without Noise}), stable limit cycles have one neutrally stable
mode that is undamped by the dynamics and appears, in the presence of noise, as a random walk along the limit cycle. The absence of a restoring force in this mode manifests
itself as an unbounded eigenvalue of the covariance matrix.

In addition to the infinite eigenvalue, there is a zero eigenvalue associated with the eigenvector
$\frac{1}{n}\left[\boldsymbol{1},0\right]^{T}$, which represents the
average amplitude of the dynamics. This indicates that near the attractor
the average amplitude is invariant to noise. This invariant feature \red{
is necessarily present wherever the dynamics are well approximated by an attractor's linear characterization}. Even in the
presence of attractor switching behavior, the average amplitude remains largely
unchanged.

The remaining \red{$2N-2$} eigenvalues and associated eigenmodes indicate the individual
character of the attractor basin each proportional to $\beta\cos(2\pi k/N)$. The average of these eigenvalues \red{$\bar{\eta}$} is given in closed form as
\red{
\begin{equation}
\label{eq:lambdabar}
\bar{\eta}=\sigma^{2}\left(1-\frac{\alpha\left(N+1\right)}{6\beta\cos\frac{2\pi k}{N}}\left[1+\Gamma(N,k,\alpha,\beta)\right]\right),
\end{equation}
where $\Gamma(N,k,\alpha,\beta)^{-1}\in16\alpha\left[\alpha,\alpha+\left|\beta\right|\right]$,} providing a metric of the attractors robustness (see Appendix \ref{App:Covariance} for details). Dependence on the wavenumber \(k\) comes in as the inverse of \(\beta\cos(2\pi k/N)\), indicating that as $\beta\cos(2\pi k/N)\rightarrow-\left|\beta\right|$,
the basins are more robust to noise. \red{Examining the metric as $N\rightarrow\infty$,
with the total input variance $\sigma_{T}^{2}=N\sigma^{2}$, wave
fraction $k_{f}=k/N$, and assuming large constant frequency $\alpha$, {\it i.e.,} $\alpha\gg1/4$
then
\[
\lim_{N\rightarrow\infty}\bar{\eta}\approx\frac{\alpha\sigma_{T}^{2}}{6\left(-\beta\cos(2\pi k_{f})\right)}.
\]
For large $N$, the attractor robustness scales linearly with oscillator
frequency and total input variance while inversely with the coupling
strength and cosine of the phase differences.}

The covariance matrix may be used to construct a quadratic quasi-potential for each rotating wave, which is guaranteed to be decreasing along the deterministic system trajectories for some finite neighborhood of the rotating wave and can be used to place lower bounds on the basin boundaries. We can build a global quasi-potential by piecewise stitching together the local potentials associated with each rotating wave, always selecting the one with the least value:

\begin{equation}
\label{eq:Potential}
V = \min_k\left(\left[\begin{array}{c}\delta a\\\delta \phi\end{array}\right]_k^T P_k \left[\begin{array}{c}\delta a\\\delta \phi\end{array}\right]_k\right).
\end{equation}

Slices of level sets of this potential for rings of \(N=3\) oscillators are plotted in Fig. \ref{fig:TimeToConvergence}. Negative coupling gives a single basin, but its locally motivated potential well is much larger than the two equal potential wells of positive coupling. To compare this to the full nonlinear system, we indicate the time-to-convergence in color, which cleanly shows the two basins of negative coupling, with the basin separatrix covering the set of points where one phase difference is zero.

\red{The covariance matrix $P_{e}$ associated with edge states $\delta e_{i}=\delta\phi_{i+1}-\delta\phi_{i}$
can be formed from the covariance matrix $P$. Due to the symmetry
in the dynamics, the diagonal elements of $P_{e}$ are in common and
correspond to the steady state variance $\bar{\sigma}^{2}$ of each edge
state with $\delta e_{i}\sim\mathcal{N}(0,\bar{\sigma}^{2})$. A probabilistic
feature that follows
is the steady state probability $p(\varepsilon_{l},\varepsilon_{u})$ that, for a single instance in time, all edge states
remain in the interval $\left[\varepsilon_{l},\varepsilon_{u}\right]$. Appendix~\ref{App:ProbabilityInterval} includes
an approximation of the probability of interval
containment using the error function $\mbox{erf}(\cdot)$, namely
\[
p(\varepsilon_{l},\varepsilon_{u})\approx\frac{1}{2^{N}}\left[\mbox{erf}\left(\frac{\varepsilon_{u}}{\bar{\sigma}\sqrt{2}}\right)-\mbox{erf}\left(\frac{\varepsilon_{l}}{\bar{\sigma}\sqrt{2}}\right)\right]^{N},
\]
where $\bar{\sigma}^{2}=\sigma^{2}\left(2-\left(4\beta N\cos(2\pi k/N)\right)^{-1}\sum_{i=1}^{N-1}\left(\alpha-\beta\cos(2\pi k/N)\sin^{2}(i\pi/N)\right)^{-1}\right)$.
Similar noise robustness characteristics can be observed over the
edge states as the full states $\delta a$ and $\delta\phi$ with
wavenumbers associated with $\beta\cos(2\pi k/N)$ close to $-\left|\beta\right|$
providing more noise robustness and so higher probabilities of maintaining
interval containment. Extending this concept into the time domain,
the probability $p_{\left[t_{1},t_{2}\right]}(\varepsilon_{l},\varepsilon_{u})$
of any edge state first exiting the interval $\left[\varepsilon_{l},\varepsilon_{u}\right]$
in time span $\left[t_{1},t_{2}\right]$ given a sampling interval
$\Delta t$ and the expected exit time $\mathbb{E}_{T}(\varepsilon_{l},\varepsilon_{u})$
of this interval are
\begin{equation}
\label{eq:ProbabilityInterval}
p_{\left[t_{1},t_{2}\right]}(\varepsilon_{l},\varepsilon_{u})=p(\varepsilon_{l},\varepsilon_{u})^{\left\lfloor t_{1}/\Delta t\right\rfloor }-p(\varepsilon_{l},\varepsilon_{u})^{\left\lfloor t_{2}/\Delta t\right\rfloor },\,\mbox{ and }\,\mathbb{E}_{T}(\varepsilon_{l},\varepsilon_{u})=-\frac{\Delta t}{\log(p(\varepsilon_{l},\varepsilon_{u}))}.
\end{equation}}

\subsection{Switching Dynamics: Phase Crossing}
\label{Sec:Phase Crossing}

\begin{figure}[t]
	\centering
	\includegraphics[width=0.8\columnwidth]{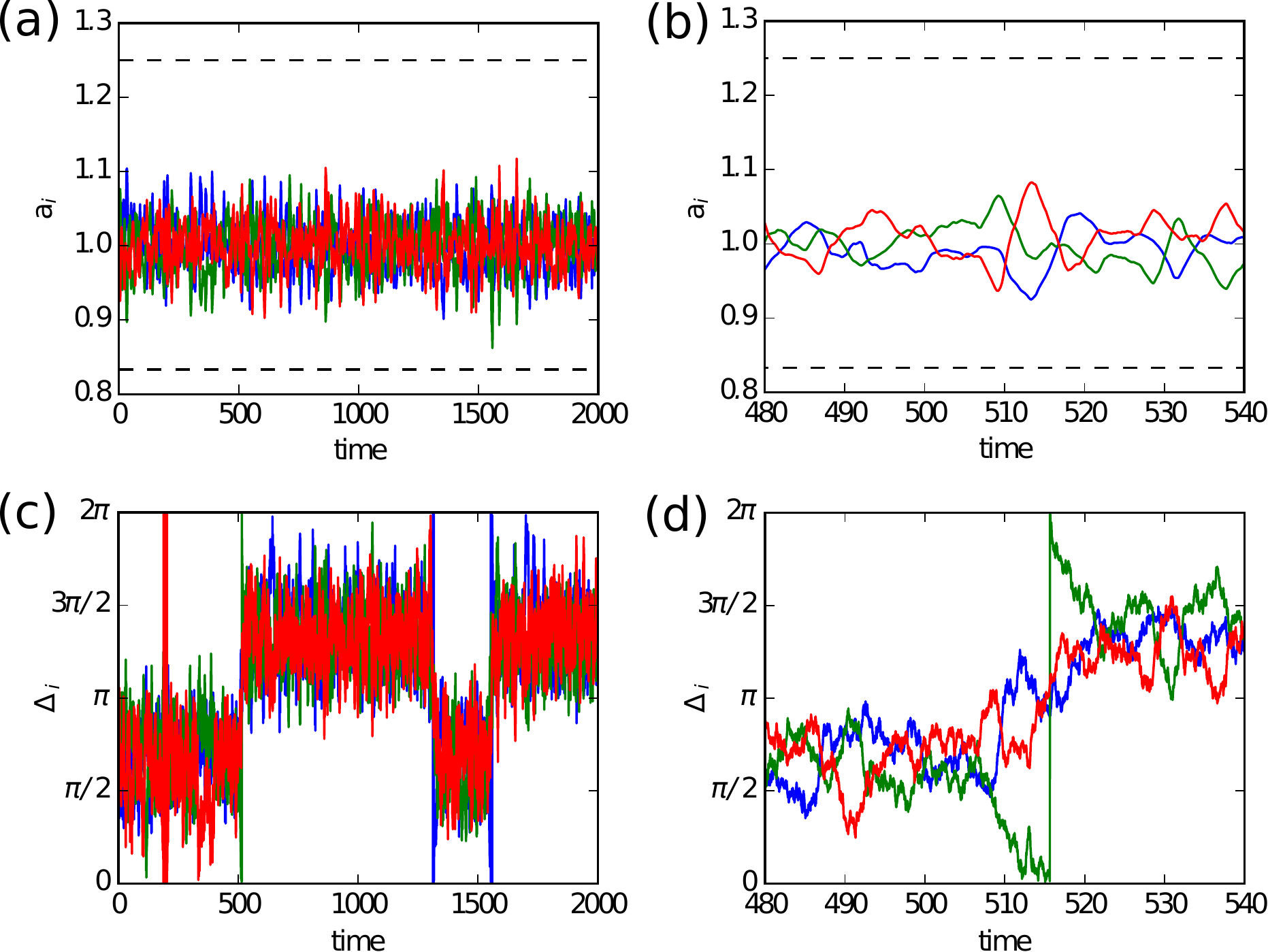}
	\caption{{\bf Attractor switching trajectories.} A representative stochastic trajectory of the \(N=3\) oscillator ring prepared with \(k=1\).
	Subfigures (a) and (b) show the three amplitudes at two different time scales. The dashed lines in these plots indicate the amplitude bounds established in Appendix \ref{App:AttractorAmplitudes}.
	Subfigure (c) shows the phase difference across each of the three edges, showing groupings at \(2\pi/3\) and \(4\pi/3\), with rapid switches between them.
	Subfigure (d) shows those same phase differences over the same time as subfigure (b), revealing that one phase difference passed through zero and rejoined the others at the other attracting state, indicating a switching event.
	This trajectory was generated with $\sigma^2=0.05$, $\beta=0.1$, and $\alpha=1$. }
	\label{fig:SwitchingTraj}
\end{figure}

So far we investigated the local properties of attractors and the response to small noise such that the system remains in the vicinity of stable rotating-wave attractors. In this section, we consider larger noise levels \red{in Eq. (\ref{eq:phaseWithNoise})} at which the system occasionally switches from the vicinity of one attractor to the vicinity of another.
Our goal is to provide a coarse-grained description of the global dynamics; \red{we wish to define an attractor switching network (ASN) in which each node represents the neighborhood of an attractor and the links connecting nodes represent the switches. To build an ASN, we must first be able to distinguish the vicinities of distinct attractors. It is computationally infeasible to capture the precise deterministic basins of attraction, so we investigate the characteristics of an attractor switch in a network of \(N=3\) oscillators to motivate a coarse-graining.}
Throughout, we employ numerical simulations using a fourth order Runge-Kutta algorithm with timestep \(t_\T{step} = 0.01\). At the end of each Runge-Kutta step, we add a zero-mean, normally distributed random number with variance \(\sigma^2 t_\T{step}\) to the phase of each oscillator \red{to capture the stochasticity of Eq.~(\ref{eq:phaseWithNoise})}.

Figure~\ref{fig:SwitchingTraj} plots a typical stochastic trajectory in a ring of $N=3$ oscillators with coupling \(\beta = 0.1\), nonlinearity \(\alpha = 1\), and noise level $\sigma^2=0.05$. As discussed in Sec.~\ref{Sec:Stability}, positive $\beta$ on the three-oscillator ring supports two stable attractors: rotating waves with wavenumbers $k = 1$ and $k=2$ (phase differences $\Delta = 2\pi/3$ and $\Delta = 4\pi/3$). Figures~\ref{fig:SwitchingTraj}a-b show the amplitudes of the three oscillators at different temporal resolutions; although noise is only directly added to the phases of the oscillators, it causes fluctuations in the amplitudes through the deterministic dynamics. However, as shown in Appendix~\ref{App:AttractorAmplitudes}, the amplitudes remain bounded within $\left[1/(1+2\left|\beta\right|),1/(1-2\left|\beta\right|)\right]$ (dashed lines in Figs.~\ref{fig:SwitchingTraj}a-b). Figure~\ref{fig:SwitchingTraj}c shows the phase differences $\Delta_i = \phi_{i}-\phi_{i-1}$ over time. The phase differences initially fluctuate around $\Delta = 2\pi/3$ and at the time of the first switch (\(\red{t}\sim 515\)) they rapidly reorganize around $\Delta = 4\pi/3$. Figure~\ref{fig:SwitchingTraj}d zooms in on that first switch, revealing that one of the phase differences passes through $0$. Indeed, such phase crossing necessarily happens if the system transitions from one rotating wave to another with a different wavenumber. Thus, the mechanism underlying the switching dynamics is associated with the phases of two neighboring oscillators crossing.

\subsection{Patterns of Patterns of Synchronization: the ASN}
\label{Sec:The ASN}

\begin{figure}[t]
	\centering
	\includegraphics[width=\columnwidth]{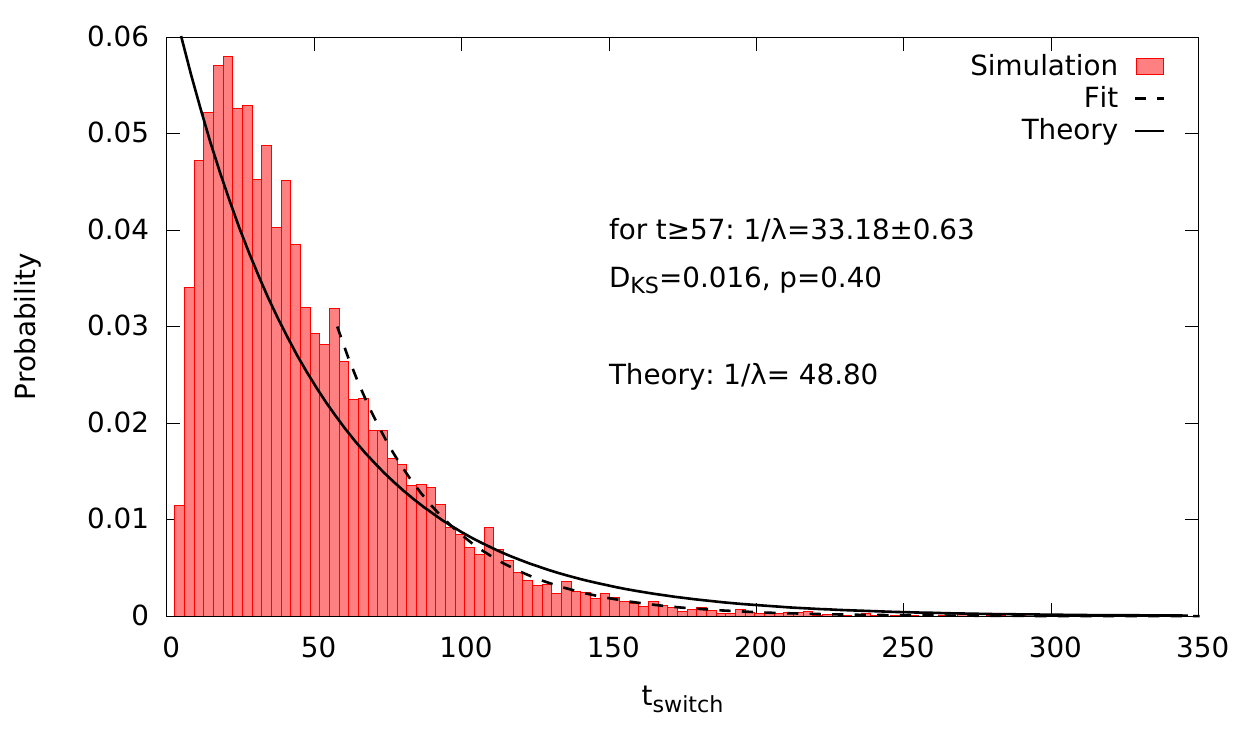}
	\caption{{\bf Switch time histogram.} Distribution of the time needed to leave state \red{$k=50$} for a ring of \red{$N=101$} oscillators based on \red{$10,000$} independent measurements with average \red{$\avg{t_\T{switch}}=47.93\pm 0.36$} where the error is the standard error of the mean. The distribution has an exponential tail:  According to the Kolmogorov-Smirnov test for \red{$t_\T{switch}\geq 57$} the distribution is consistent with an exponential distribution with \red{$1/\lambda=33.18\pm 0.63$ ($D_\T{KS}=0.016$, $p$-value $0.40$)}, where the $\lambda$ is the maximum likelihood fit of the rate parameter and the error \red{corresponds to the 95\% confidence interval}. The variance of the noise is \red{$\sigma^2=0.1$}; $\beta=0.1$ and  $\alpha=1$. \red{The theoretical curve is the probability of a zero crossing $p_{\left[t_{1},t_{2}\right]}(-2\pi k/N,2\pi-2\pi k/N)$ based on the linear analysis.}}
	\label{fig:switchtimehist}
\end{figure}

\begin{figure}[t]
	\centering
	\includegraphics[width=\columnwidth]{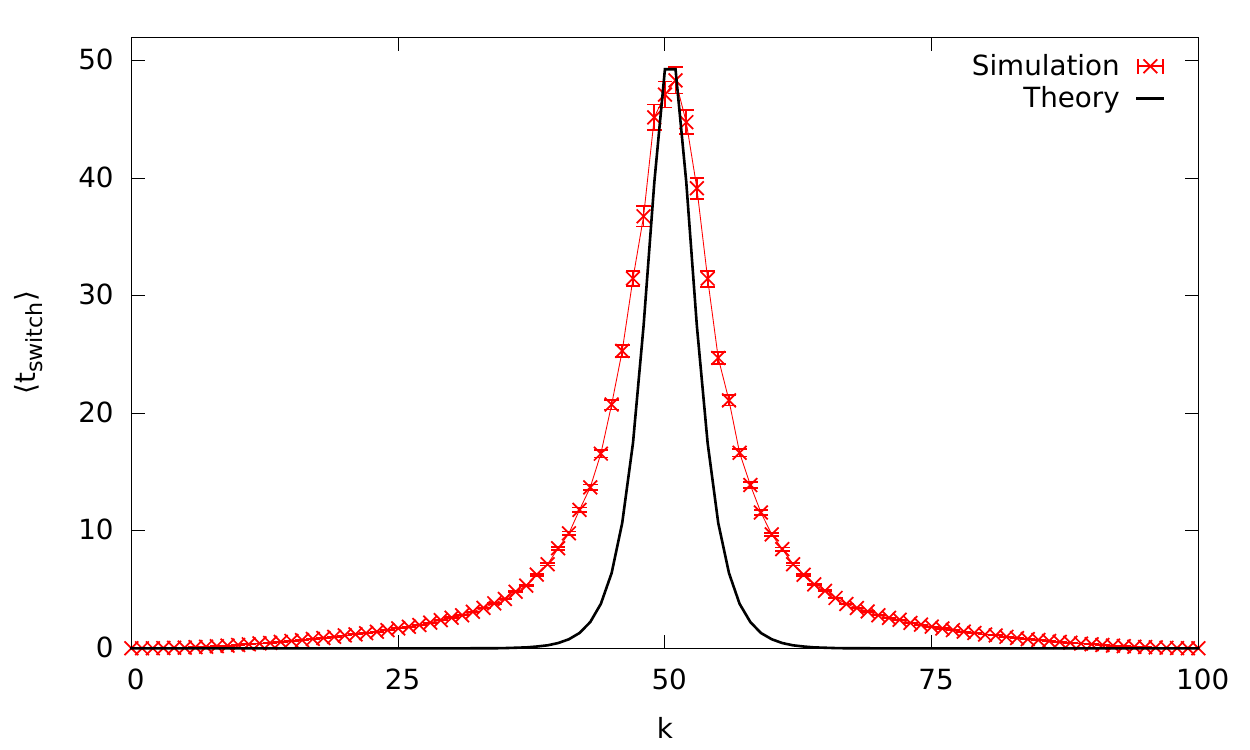}
	\caption{{\bf Average switch time.} Average switch time in function of the wave number $k$ indexing the limit cycles states of a ring of $N=101$ oscillators. Each point is an average over $1,000$ independent measurements, and the error bars represent the standard error of the mean. These simulations were run with $\sigma^2=0.1$, $\beta=0.1$, and $\alpha=1$. \red{The solid black curve shows the analytic prediction of the expected zero crossing time $\mathbb{E}_{T}(-2\pi k/N,2\pi-2\pi k/N)$.}}
	\label{fig:switchtime-k}
\end{figure}

\begin{figure*}[t]
	\centering
	\includegraphics[width=\textwidth]{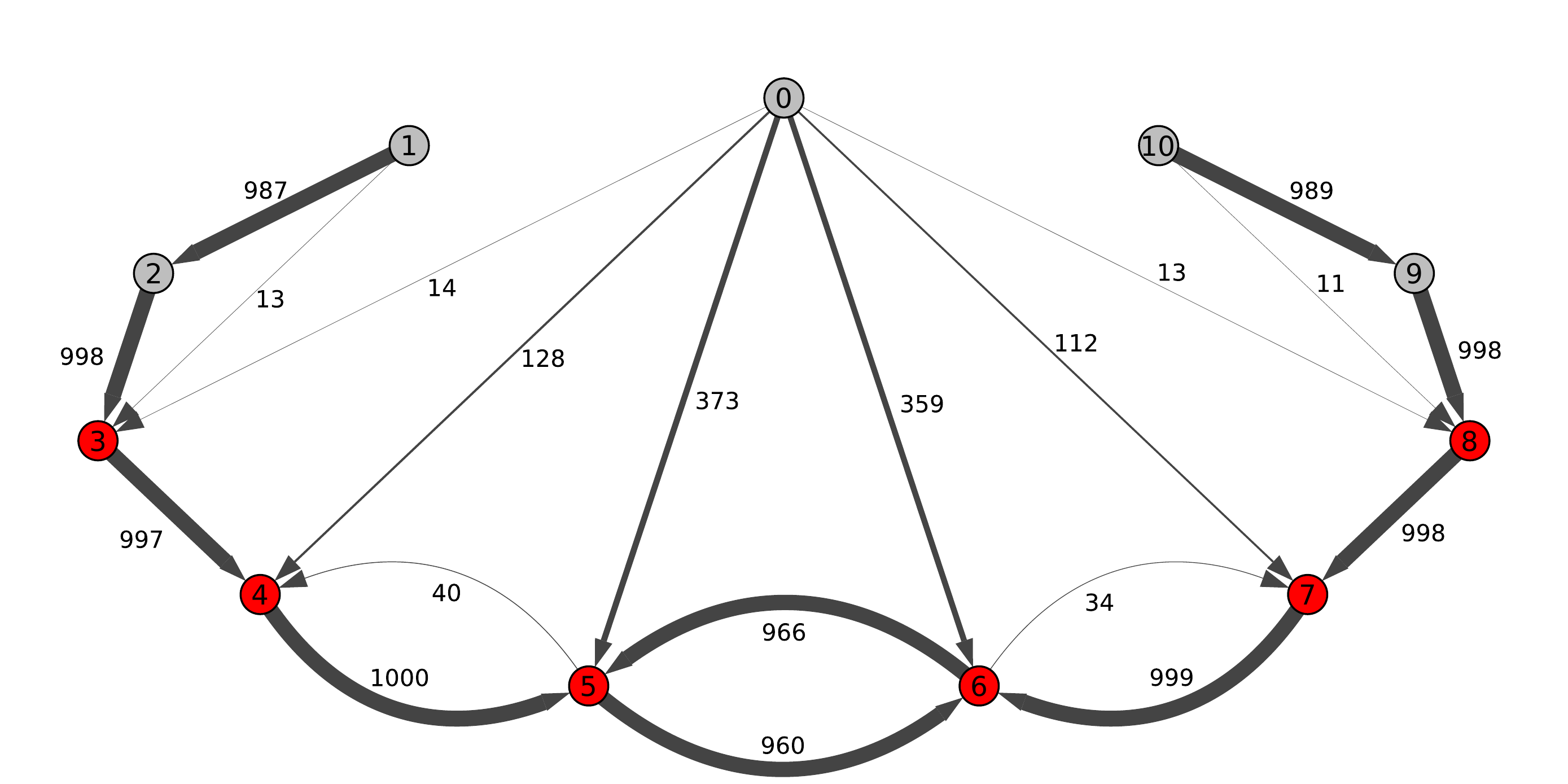}
	\caption{{\bf Attractor switching network.} Each node represents a limit cycle state of a ring of $N=11$ oscillators; grey nodes represent unstable limit cycles, red nodes represent stable limit cycles. The system is prepared in each $k$ state $1,000$ times and we record to which state it switches. The labels and the width of the links represent the transition counts. Links with less than $5$ transitions are not shown. These simulations were run with $\sigma^2=0.1$, $\beta=0.1$, and $\alpha=1$.}
	\label{fig:asn}
\end{figure*}

\red{Finally, we partition the state space into regions enclosing each limit cycle according to wavenumber \(k\) and investigate attractor switching phenomena as characterized by these partition boundaries. In particular, we study the distribution of time needed to escape an attractor, the average times for such a switch to occur, and the overall organization of the attractor switching network (ASN).}

Driven by the observation that attractor switching is accompanied by a phase crossing, we choose to identify a switch as an event when any $\Delta_i$ becomes $0$. More precisely we calculate
\begin{equation}
\label{eq:state_detector}
k = \frac{\sum_i \Delta_i}{2\pi},
\end{equation}
where $\Delta_i\in[0,2\pi)$. Since the oscillators are organized in a ring, $k$ is an integer. If the system is on a deterministic attractor, $k$ is equal to the corresponding wavenumber. Thus, $k$ changes value only when a $\Delta_i$ passes through zero. We therefore coarse grain the state space by assigning the system to be in rotating-wave ``state'' $k$ as defined by Eq.~(\ref{eq:state_detector}). \red{The magnitude of change in \(k\) is precisely equal to the number of adjacent phase difference that pass through zero at a particular time.} Note that this assigns different volumes of state space to different rotating wave states. For example, $k=0$ only if all $\Delta_i=0$, and small fluctuations in the phase differences cause discrete fluctuations in $k$. Hence, this choice of coarse-graining is natural only if $k=0$ is unstable, \textit{i.e.,} $\beta>0$.

We perform measurements of switching by preparing the system in a stable attractor of the deterministic dynamics, letting it evolve until the system switches to another state according to Eq.~(\ref{eq:state_detector}), and then recording the time taken to switch, $t_\T{switch}$, and the new state. In Fig.~\ref{fig:switchtimehist}, we show a histogram of $t_\T{switch}$ based on \red{$10,000$} independent runs for state \red{$k=50$} of a ring of \red{$N=101$} oscillators with $\beta=0.1$ and noise level $\sigma^2=0.1$. We find that the tail of the histogram is consistent with an exponential distribution; the typical time needed to switch is therefore well characterized by the average $\avg{t_\T{switch}}$. \red{The linear analysis prediction of switching probabilities is described
by $p_{\left[t_{1},t_{2}\right]}(\varepsilon_{l},\varepsilon_{u})$
in Eq.~(\ref{eq:ProbabilityInterval}) with the zero cross condition corresponding to $\left[\varepsilon_{l},\varepsilon_{u}\right]=\left(-2\pi k/N,2\pi-2\pi k/N\right)$.
A similar exponential tail is noted between between both curves. The
discrepancy for small $\avg{t_\T{switch}}$,
can be attributed to the linear regime assumption within the $p_{\left[t_{1},t_{2}\right]}(\varepsilon_{l},\varepsilon_{u})$
calculation, specifically the independence of edge states over time.
For small $\avg{t_\T{switch}}$, the simulation
is exhibiting a distribution similar to the hitting time induced by
Brownian motion rather than the independent and identically distributed random variable sampling of the linear analysis.}

In Fig.~\ref{fig:switchtime-k} we show $\avg{t_\T{switch}}$ as a function of $k$ for $N=101$, $\beta=0.1$ and $\sigma^2=0.1$. We find that $\avg{t_\T{switch}}$ is sharply peaked at $k=50$, and vanishes as the system approaches the fully synchronized state $k=0$ or, equivalently, $k=101$. \red{We compare the nonlinear stability measure $\avg{t_\T{switch}}$
to the expected switching time $\mathbb{E}_{T}(\varepsilon_{l},\varepsilon_{u})$ based on the steady state covariance,
where $\left[\varepsilon_{l},\varepsilon_{u}\right]=\left(-2\pi k/N,2\pi-2\pi k/N\right)$,
defined by Eq.~(\ref{eq:ProbabilityInterval}). The general shape and scale of the curves agree
with deviations occurring as the curves depart from $k=50$. As the
dynamics are unstable about the rotating-wave states $k\notin\left[26,75\right]$,
the linear analysis indicates an instantaneous switch compared to the nonlinear case where some time is required to depart from the unstable limit cycle. Deviations
in the stable regime $k\in\left[26,75\right]$, can be attributed
to uncaptured higher order modes in the dynamics and variable size of the linear regime across $k$.}

Finally, we construct the ASN for a ring of $N=11$ oscillators, with $\beta=0.1$ and $\sigma^2=0.1$ by preparing the system in each $k$ rotating-wave state $1,000$ times and recording to which state it switches. We show the ASN in Fig.~\ref{fig:asn}; red nodes represent stable rotating-wave states and gray nodes unstable states. We draw a directed link from node $k_1$ to node $k_2$ if we observed a switch from $k_1$ to $k_2$. The link weight is the count of observed switches. It is unlikely that two $\Delta_i$'s become zero simultaneously, therefore typically switching happens from state $k$ to neighboring states $k\pm 1$. The most unstable state $k=0$ is an exception, because at $k=0$ each $\Delta_i=0$ and this allows switching to any state. \red{In the few other cases where this occurs, the system simply passed through the intermediate partitions within a single time step of simulation. That is, multiple \(\Delta_i\)'s became zero within one \(t_\T{step}\) increment.}  Overall the system evolves towards states where the adjacent oscillators are most out of phase, $k=5$ and $k=6$, and it rarely leaves these states.

\red{ Although we have not proven that our list of limit cycles captures all attractors of the deterministic system, the lack of cycles with low \(\avg{t_\T{switch}}\) in the ASN provide indication that any further attractors are contained within a single partition and are therefore associated with a single wavenumber. }

This example demonstrates that dynamical coarse-graining of the state space is an informative and necessary approach when constructing attractor switching networks for systems with noisy dynamics. Moreover, ASNs provide an insightful description of the complex and high-dimensional dynamics of noisy, multistable systems.

\section{Conclusion}
\label{sec:Conc}

\red{Our long-term goal is to understand the architecture of basins of attraction in large-scale complex dynamical systems and to develop methods that reveal how state-space structures can facilitate driving between basins. Here, we took several key steps toward these larger goals by analyzing in-depth synchronization phenomena in a system of coupled oscillators arranged in a ring topology.}  From the equations governing the evolution of the system, we first predict analytically the different patterns of synchronization that can exist (i.e., rotating wave solutions) and analyze their local stability via the linearization of the governing equations. We then analyze the covariance matrix of deviations around the attracting rotating waves and use this to construct a piecewise quadratic quasi-potential roughly describing the full attractor space. \red{We additionally use this covariance matrix to make predictions about the fluctuations of phase differences.} Although the covariance analysis allows us to analytically calculate a metric for the robustness of each attractor to noise, we turn to simulation to deal with the impact of large noise.  With this, we can explore the mechanisms associated with attractor switching and develop the attractor switching network. Doing so reveals a clear and strong drive towards those rotating waves with wavenumber approximately half the number of nodes, such that adjacent oscillators are nearly out of phase.

The techniques developed here should generalize to other systems and provide an systematic and analytic advance for developing the underlying theory of attractor switching networks.  We intend to further this study by carefully investigating the dynamics of single switches in larger rings, extending our methods to complex networks with richer attractor types, and validating them in NEMS nanoscale device experiments.

\section{Acknowledgements}
\label{sec:Ack}
We thank Mike Cross, Leonardo Due{\~n}as-Osorio, Warren Fon, Matt Matheny, Michael Roukes, and Sam Stanton for helpful discussions. This material is based upon work supported by, or in part by, the U.S. Army Research Laboratory and the U. S. Army Research Office under Multidisciplinary University Research Initiative award W911NF-13-1-0340.

\newpage
\newpage
\newpage
\newpage
\newpage
\newpage
\newpage
\newpage
\newpage
\newpage
\newpage
\newpage
%##############################################%
%##############################################%
%##############################################%
%##############################################%
%##############################################%
%##############################################%
%##############################################%%##############################################%
%##############################################%
%##############################################%
%##############################################%
%##############################################%
%##############################################%

%\bibliographystyle{unsrt}
%\bibliography{NEMSChaosSpecialIssue3_Initials}

\begin{thebibliography}{10}

\bibitem{davidson2005gene}
E~Davidson and M~Levin.
\newblock Gene regulatory networks.
\newblock {\em Proceedings of the National Academy of Sciences of the United
  States of America}, 102(14):4935--4935, 2005.

\bibitem{gray1992synchronization}
CM~Gray, AK~Engel, P~K{\"o}nig, and W~Singer.
\newblock Synchronization of oscillatory neuronal responses in cat striate
  cortex: temporal properties.
\newblock {\em Visual Neuroscience}, 8(04):337--347, 1992.

\bibitem{rinaldi2001identifying}
SM~Rinaldi, JP~Peerenboom, and TK~Kelly.
\newblock Identifying, understanding, and analyzing critical infrastructure
  interdependencies.
\newblock {\em Control Systems, IEEE}, 21(6):11--25, 2001.

\bibitem{rosato2008modelling}
V~Rosato, L~Issacharoff, F~Tiriticco, S~Meloni, S~Porcellinis, and R~Setola.
\newblock Modelling interdependent infrastructures using interacting dynamical
  models.
\newblock {\em International Journal of Critical Infrastructures},
  4(1-2):63--79, 2008.

\bibitem{gonzalez2015interdependent}
AD~Gonz{\'a}lez, L~Due{\~n}as-Osorio, M~S{\'a}nchez-Silva, and AL~Medaglia.
\newblock The interdependent network design problem for optimal infrastructure
  system restoration.
\newblock {\em Computer-Aided Civil and Infrastructure Engineering}, 2015.

\bibitem{jpc8}
H Poincar{\'e}.
\newblock Les nouvelles m{\'e}thodes de la m{\'e}canique c{\'e}leste.
\newblock {\em Gauthier-Villars, Paris}, 1892.

\bibitem{jpc6}
J~Milnor.
\newblock On the concept of attractor.
\newblock In {\em The Theory of Chaotic Attractors}, pages 243--264. Springer,
  1985.

\bibitem{jpc7}
S~Smale.
\newblock {\em The Mathematics of Time}.
\newblock Springer, 1980.

\bibitem{jpc4}
C~Conley.
\newblock The gradient structure of a flow: I.
\newblock {\em Ergodic Theory and Dynamical Systems}, 8(8*):11--26, 1988.

\bibitem{jpc5}
C~Conley.
\newblock {\em Isolated invariant sets and the Morse index}.
\newblock Number~38. American Mathematical Soc., 1978.

\bibitem{jpc3}
DE~Norton.
\newblock The fundamental theorem of dynamical systems.
\newblock {\em Commentationes Mathematicae Universitatis Carolinae},
  36(3):585--598, 1995.

\bibitem{jpc1}
J~Hasty, J~Pradines, M~Dolnik, and JJ~Collins.
\newblock Noise-based switches and amplifiers for gene expression.
\newblock {\em Proceedings of the National Academy of Sciences},
  97(5):2075--2080, 2000.

\bibitem{jpc2}
M~Timme, F~Wolf, and T~Geisel.
\newblock Prevalence of unstable attractors in networks of pulse-coupled
  oscillators.
\newblock {\em Physical review letters}, 89(15):154105, 2002.

\bibitem{PIK03}
A~Pikovsky, M~Rosenblum, and J~Kurths.
\newblock {\em Synchronization: a universal concept in nonlinear sciences},
  volume~12.
\newblock Cambridge University Press, 2003.

\bibitem{HOP97}
FC~Hoppensteadt and EM~Izhikevich.
\newblock {\em Weakly connected neural networks}.
\newblock Springer, Berlin, 1997.

\bibitem{MOT13}
AE~Motter, SA~Myers, M~Anghel, and T~Nishikawa.
\newblock Spontaneous synchrony in power-grid networks.
\newblock {\em Nature Physics}, 9(3):191--197, 2013.

\bibitem{NED00}
Z~N{\'e}da, E~Ravasz, Y~Brechet, T~Vicsek, and A-L Barab{\'a}si.
\newblock Self-organizing processes: The sound of many hands clapping.
\newblock {\em Nature}, 403(6772):849--850, 2000.

\bibitem{MIR90}
RE~Mirollo and SH~Strogatz.
\newblock Synchronization of pulse-coupled biological oscillators.
\newblock {\em SIAM Journal on Applied Mathematics}, 50(6):1645--1662, 1990.

\bibitem{MAT91}
PC~Matthews, RE~Mirollo, and SH~Strogatz.
\newblock Dynamics of a large system of coupled nonlinear oscillators.
\newblock {\em Physica D: Nonlinear Phenomena}, 52(2):293--331, 1991.

\bibitem{ABR04}
DM~Abrams and SH~Strogatz.
\newblock Chimera states for coupled oscillators.
\newblock {\em Physical Review Letters}, 93(17):174102, 2004.

\bibitem{HAG12}
AM~Hagerstrom, TE~Murphy, R~Roy, P~H{\"o}vel, I~Omelchenko, and E~Sch{\"o}ll.
\newblock Experimental observation of chimeras in coupled-map lattices.
\newblock {\em Nature Physics}, 8(9):658--661, 2012.

\bibitem{ARE06}
A~Arenas, A~D{\'\i}az-Guilera, and CJ~P{\'e}rez-Vicente.
\newblock Synchronization reveals topological scales in complex networks.
\newblock {\em Physical Review Letters}, 96(11):114102, 2006.

\bibitem{WAN09}
K~Wang, X~Fu, and K~Li.
\newblock Cluster synchronization in community networks with nonidentical
  nodes.
\newblock {\em Chaos: An Interdisciplinary Journal of Nonlinear Science},
  19(2):023106, 2009.

\bibitem{ERM85}
GB~Ermentrout.
\newblock The behavior of rings of coupled oscillators.
\newblock {\em Journal of Mathematical Biology}, 23(1):55--74, 1985.

\bibitem{BRE97}
PC~Bressloff, S~Coombes, and B~De~Souza.
\newblock Dynamics of a ring of pulse-coupled oscillators: Group-theoretic
  approach.
\newblock {\em Physical Review Letters}, 79(15):2791, 1997.

\bibitem{ROG04}
JA~Rogge and D~Aeyels.
\newblock Stability of phase locking in a ring of unidirectionally coupled
  oscillators.
\newblock {\em Journal of Physics A: Mathematical and General}, 37(46):11135,
  2004.

\bibitem{WIL06}
DA~Wiley, SH~Strogatz, and M~Girvan.
\newblock The size of the sync basin.
\newblock {\em Chaos: An Interdisciplinary Journal of Nonlinear Science},
  16(1):015103, 2006.

\bibitem{CHO10}
C-U Choe, T~Dahms, P~H{\"o}vel, and E~Sch{\"o}ll.
\newblock Controlling synchrony by delay coupling in networks: from in-phase to
  splay and cluster states.
\newblock {\em Physical Review E}, 81(2):025205, 2010.

\bibitem{HA12}
S-Y Ha and M-J Kang.
\newblock On the basin of attractors for the unidirectionally coupled kuramoto
  model in a ring.
\newblock {\em SIAM Journal on Applied Mathematics}, 72(5):1549--1574, 2012.

\bibitem{SHA15}
AV~Shabunin.
\newblock Phase multistability in a dynamical small world network.
\newblock {\em Chaos: An Interdisciplinary Journal of Nonlinear Science},
  25(1):013109, 2015.

\bibitem{LIF08}
R~Lifshitz and MC~Cross.
\newblock Nonlinear dynamics of nanomechanical and micromechanical resonators.
\newblock {\em Review of Nonlinear Dynamics and Complexity}, 1:1--52, 2008.

\bibitem{CRO04}
MC~Cross, A~Zumdieck, R~Lifshitz, and JL~Rogers.
\newblock Synchronization by nonlinear frequency pulling.
\newblock {\em Physical Review Letters}, 93(22):224101, 2004.

\bibitem{MAT14}
MH~Matheny, M~Grau, LG~Villanueva, RB~Karabalin, MC~Cross, and ML~Roukes.
\newblock Phase synchronization of two anharmonic nanomechanical oscillators.
\newblock {\em Physical Review Letters}, 112(1):014101, 2014.

\bibitem{fonUpcoming}
W~Fon, M~Matheny, J~Li, RM~D'Souza, JP~Crutchfield, and ML~Roukes.
\newblock Modular nonlinear nanoelectromechanical oscillators for synchronized
  networks, under review.

\bibitem{kuramoto1984chemical}
Y~Kuramoto.
\newblock Chemical oscillations, waves and turbulence, 1984.

\bibitem{ARA02}
IS~Aranson and L~Kramer.
\newblock The world of the complex ginzburg-landau equation.
\newblock {\em Reviews of Modern Physics}, 74(1):99, 2002.

\bibitem{COR13}
SP~Cornelius, WL~Kath, and AE~Motter.
\newblock Realistic control of network dynamics.
\newblock {\em Nature Communications}, 4, 2013.

\bibitem{PIS14}
AN~Pisarchik and U~Feudel.
\newblock Control of multistability.
\newblock {\em Physics Reports}, 540(4):167--218, 2014.

\bibitem{WEL15}
DK~Wells, WL~Kath, and AE~Motter.
\newblock Control of stochastic and induced switching in biophysical networks.
\newblock {\em Physical Review X}, 5(3):031036, 2015.

\bibitem{Khalil1996}
HK~Khalil.
\newblock {\em {Nonlinear Systems}}.
\newblock Prentice Hall, Upper Saddle River, 1996.

\bibitem{Diestel2000}
R~Diestel.
\newblock {\em {Graph Theory}}.
\newblock Springer, Berlin, 2000.

\bibitem{Horn1990}
RA~Horn and CR~Johnson.
\newblock {\em {Matrix Analysis}}.
\newblock Cambridge University Press, New York, 1990.

\bibitem{Skogestad2005}
S~Skogestad and I~Postlethwaite.
\newblock {\em {Multivariable Feedback Control: Analysis and Design}}.
\newblock Wiley, West Sussex, 2005.

\end{thebibliography}

\appendix
\section{Amplitude Bounds on Attractors}
\label{App:AttractorAmplitudes}

Consider the quasi-potential $V=\sum_{i=1}^{n}\left|a_{i}-1\right|$
and assume that $\left|\beta\right|<1/2$ then
\begin{eqnarray*}
\frac{dV}{dt} & = & \sum_{i=1}^{n}\mbox{sgn}(a_{i}-1)\dot{a}_{i}\\
 & = & \sum_{i=1}^{n}\mbox{sgn}(a_{i}-1)\left[-\frac{1}{2}\left(a_{i}-1\right)+\frac{\beta}{2}\Big(a_{i+1}\sin\left(\phi_{i+1}-\phi_{i}\right)+a_{i-1}\sin\left(\phi_{i-1}-\phi_{i}\right)\Big)\right]\\
 & \leq & \sum_{i=1}^{n}-\frac{1}{2}\mbox{sgn}(a_{i}-1)\left(a_{i}-1\right)+\frac{\left|\beta\right|}{2}\big(\left|a_{i+1}\right|+\left|a_{i-1}\right|\big)\\
 & = & \sum_{i=1}^{n}-\frac{1}{2}\left|a_{i}-1\right|+\left|\beta\right|\left|a_{i}\right|
\end{eqnarray*}
and so for $\left\Vert a-\boldsymbol{1}\right\Vert _{1}\geq2\left|\beta\right|\left\Vert a\right\Vert _{1}$
then $dV/dt\leq0$. Hence, the dynamics will converge to the invariant
set $\mathcal{B}=\left\{ a|\left\Vert a-\boldsymbol{1}\right\Vert _{1}\leq2\left|\beta\right|\left\Vert a\right\Vert _{1}\right\} $
\cite{Khalil1996}. The smallest annulus containing $\mathcal{B}$
is $a_{i}\in\left[\frac{1}{1+2\left|\beta\right|},\frac{1}{1-2\left|\beta\right|}\right]$
and so the dynamics will converge to this annulus.

\section{Linearization Stability Equivalence}
\label{App:Linearization}

The linearized dynamics state matrix can be formalized as a series
of Kronecker sums as

\[
F=\frac{1}{2}\left(\left[\begin{array}{cc}
-1 & 0\\
4\alpha & 0
\end{array}\right]\otimes I+\left[\begin{array}{cc}
0 & \beta c\\
-\beta c & 0
\end{array}\right]\otimes L+\left[\begin{array}{cc}
-\beta s & 0\\
0 & -\beta s
\end{array}\right]\otimes M\right),
\]
where $c=\cos\left(2\pi k/N\right)$ and $s=\sin\left(2\pi k/N\right)$.
Now $\left[\begin{array}{c}
0\\
1
\end{array}\right]\otimes\boldsymbol{1}/\sqrt{n}$ is a right eigenvector of $F$, with associated left eigenvector
$\left[\begin{array}{c}
4\alpha\\
1
\end{array}\right]\otimes\boldsymbol{1}/\sqrt{n}$ and unique eigenvalue 0. This follows from $L\boldsymbol{1}=\boldsymbol{1}^{T}L=0$
and $M\boldsymbol{1}=\boldsymbol{1}^{T}M=0$ and by examining the
eigenvectors and eigenvalues of the matrix $\left[\begin{array}{cc}
-1 & 0\\
4\alpha & 0
\end{array}\right]$. Denoting the eigenvalues of an arbitrary matrix $Z$ as $\mu_{1}\left(Z\right),\mu_{2}(Z),\dots$,
where $\left|\mbox{Re}(\mu_{1}(Z))\right|\leq\left|\mbox{Re}(\mu_{2}(Z))\right|\leq\dots$,
then $\mu_{1}(F)=0$.

Consider the matrices
\begin{eqnarray}
F_{1} & = & F-y\left[\begin{array}{cc}
0 & 0\\
4\alpha & 1
\end{array}\right]\otimes\boldsymbol{1}\boldsymbol{1}^{T}/n,\label{eq:Perturbed Linearized Dynamics}\\
F_{2} & = & F_{1}-\frac{1}{2}\left[\begin{array}{cc}
-\beta s & 0\\
0 & -\beta s
\end{array}\right]\otimes M,\mbox{ and}\label{eq:Removed M Linearized Dynamics}\\
\widetilde{F} & = & \lim_{y\rightarrow0}F_{2},\label{eq:Perturbed Removed M}
\end{eqnarray}
then
\begin{eqnarray*}
 & \mbox{Re}(\mu_{i}(F))<0 & \mbox{ for }i\neq1\\
\mbox{\ensuremath{\iff}} & \mbox{Re}(\mu_{i}(F_{1}))<0 & \mbox{ for all }i\mbox{ (By shifting the null space associated with }\mu_{1}(F))\\
\mbox{\ensuremath{\iff}} & \mbox{Re}(\mu_{i}(F_{2}))<0 & \mbox{ for all }i\mbox{ (By Proposition~\ref{Prop: Lyapunov Equivalence})}\\
\iff & \mbox{Re}(\mu_{i}(\widetilde{F}))<0 & \mbox{ for }i\neq1\mbox{ (By shifting the null space associated with }\mu_{1}(\widetilde{F})).\\
\iff & \mbox{Re}\left(-1\pm\sqrt{16\alpha x_{i}-4x_{i}^{2}+1}\right)<0 & \mbox{ for }i\neq1\mbox{ (By Proposition~\ref{Prop: Eigenvalues of A2})}\\
\iff & 4\alpha x_{i}-x_{i}^{2}<0 & \mbox{ for }i\neq1.
\end{eqnarray*}

\begin{prop*}
\label{Prop: Lyapunov Equivalence}Consider $Q$ positive semidefinite,
$P$ positive definite, $y>0$ and Hurwitz matrix $F_{2}$. The matrix
$P$ satisfies $F_{2}P+PF_{2}^{T}=-Q\otimes I$ if and only if $F_{1}P+PF_{1}^{T}=-Q\otimes I$.\end{prop*}
\begin{proof}
Consider the permutation matrix $E\in\mathbb{R}^{n\times n}$ defined
as $E_{n,i}=1$, $E_{i,i+1}=1$ for $i=1,\dots,n-1$ and $E_{ij}=0$
otherwise. The action of the permutation matrix \(E\) on a vector corresponds
to a mapping of element $i$ to element $i-1$ (mod $n$), and corresponds
to the rotational automorphism on an $n$ node ring graph \cite{Diestel2000}.
As $E$ represents an automorphism of the graph, $ELE^{T}=L$ and $EME^{T}=M$.

From the Lyapunov equation $F_{2}P+PF_{2}^{T}=-Q\otimes I,$ as $E$
is a permutation matrix on $L$, then $\left(I\otimes E\right)F_{2}=F_{2}\left(I\otimes E\right),$
and

\begin{widetext}

\begin{eqnarray*}
\left(I\otimes E\right)\left(F_{2}P+PF_{2}^{T}\right)\left(I\otimes E^{T}\right) & = & -\left(I\otimes E\right)Q\otimes I\left(I\otimes E^{T}\right)\\
F_{2}\left(I\otimes E\right)P\left(I\otimes E^{T}\right)+\left(I\otimes E\right)P\left(I\otimes E^{T}\right)F_{2}^{T} & = & -Q\otimes I\\
F_{2}\tilde{P}+\tilde{P}F_{2}^{T} & = & -Q\otimes I.
\end{eqnarray*}

As $F_{2}$ is Hurwitz, the solution to the Lyapunov equation is unique
\cite{Horn1990}. Hence, $\tilde{P}=P$ and $\left(I\otimes E\right)P=P\left(I\otimes E\right)$
and $\left(I\otimes E^{T}\right)P=P\left(I\otimes E^{T}\right)$.
Therefore,
\[
\left(I\otimes M\right)P=\left(I\otimes\left(E-E^{T}\right)\right)P=P\left(I\otimes\left(E-E^{T}\right)\right)=P\left(I\otimes M\right),
\]
and
\begin{eqnarray*}
F_{1}^{T}P+PF_{1} & = & \left(F_{2}+\frac{1}{2}\beta sI\otimes M\right)^{T}P+P\left(F_{2}+\frac{1}{2}\beta sI\otimes M\right)\\
 & = & F_{2}^{T}P+PF_{2}+\frac{1}{2}\beta s\left[\left(I\otimes M\right)^{T}P+P\left(I\otimes M\right)\right]\\
 & = & -I+\frac{1}{2}\beta s\left[-\left(I\otimes M\right)P+\left(I\otimes M\right)P\right]=-I.
\end{eqnarray*}

\end{widetext}

\end{proof}

\begin{prop*}
\label{Prop: Eigenvalues of A2}The eigenvectors of $\widetilde{F}$
are of the form $v_{1i}\otimes w_{i}$ and $v_{2i}\otimes w_{i}$
where $w_{i}$ is a unit eigenvector of $L$ and $v_{1i}$ and $v_{2i}$
are the eigenvectors of the matrix
\begin{equation}
\frac{1}{2}\left(\left[\begin{array}{cc}
-1 & 0\\
4\alpha & 0
\end{array}\right]+x_{i}\left[\begin{array}{cc}
0 & 1\\
-1 & 0
\end{array}\right]\right).\label{eq:Base eigenvector matrix}
\end{equation}
The associated eigenvalues of \(\tilde{F}\) are $\mu_{11}=0$, $\mu_{21}=-\frac{1}{2}$,
and
\begin{equation}
\mu_{1i,2i}=\frac{1}{4}\left(-1\pm\sqrt{16\alpha x_{i}-4x_{i}^{2}+1}\right)\label{eq:Eigenvalues in closed form}
\end{equation}
for $i\neq1$. Here, $x_{i}=\beta c\lambda_{i}$ where $0\leq\lambda_{1}\leq\dots\leq\lambda_{n}$
are the eigenvalues of $L$.\end{prop*}
\begin{proof}
This result follows as
\begin{eqnarray*}
\widetilde{F}(v_{1i}\otimes w_{i}) & = & \frac{1}{2}\left(\left[\begin{array}{cc}
-1 & 0\\
4\alpha & 0
\end{array}\right]\otimes I+\left[\begin{array}{cc}
0 & \beta c\\
-\beta c & 0
\end{array}\right]\otimes L\right)v_{1i}\otimes w_{i}\\
 & = & \frac{1}{2}\left(\left[\begin{array}{cc}
-1 & 0\\
4\alpha & 0
\end{array}\right]v_{1i}\otimes w_{i}+\left[\begin{array}{cc}
0 & 1\\
-1 & 0
\end{array}\right]v_{1i}\otimes\left(\beta cL\right)w_{i}\right)\\
 & = & \frac{1}{2}\left(\left[\begin{array}{cc}
-1 & 0\\
4\alpha & 0
\end{array}\right]v_{1i}\otimes w_{i}+x_{i}\left[\begin{array}{cc}
0 & 1\\
-1 & 0
\end{array}\right]v_{1i}\otimes w_{i}\right)\\
 & = & \frac{1}{2}\left(\left[\begin{array}{cc}
-1 & 0\\
4\alpha & 0
\end{array}\right]+x_{i}\left[\begin{array}{cc}
0 & 1\\
-1 & 0
\end{array}\right]\right)v_{1i}\otimes w_{i}=\mu_{1i}v_{1i}\otimes w_{i}.
\end{eqnarray*}
Hence, by solving for the eigenvalues of matrix (\ref{eq:Base eigenvector matrix}),
the eigenvalues of $\widetilde{F}$ in closed form follow.
\end{proof}

\section{Small Noise Covariance}
\label{App:Covariance}

The covariance matrix $P'$ associated with noise driven dynamics
(\ref{eq:Linearized Dynamics with Noise}) with $w\sim\mathcal{N}(0,\sigma^{2}I)$
can be found using the Lyapunov equation
\[
F_{1}^{T}P+PF_{1}+Q\otimes I=0,
\]
where $Q=\left[\begin{array}{cc}
0 & 0\\
0 & \sigma^{2}
\end{array}\right]$ \cite{Skogestad2005} and then taking the limit of $P'=\lim_{y\rightarrow0}P$
. From Proposition \ref{Prop: Lyapunov Equivalence}, $P$ satisfies $F_{2}^{T}P+PF_{2}=-Q\otimes I.$ Let $F_{2}=\left(V\otimes W\right)\Lambda\left(V^{-1}\otimes W^{-1}\right)$
where $V\otimes W$ represents the eigenvectors of $F_{2}$ and $\Lambda$
the diagonal matrix of its eigenvalues. Further, as $L$ is symmetric
then $WW^{T}=I$. Hence,
\begin{eqnarray*}
0 & = & F_{2}^{T}P+PF_{2}+Q\otimes I\\
 & = & \left(V^{-T}\otimes W^{-T}\right)\Lambda\left(V^{T}\otimes W^{T}\right)P+P\left(V\otimes W\right)\Lambda\left(V^{-1}\otimes W^{-1}\right)+Q\otimes I.
\end{eqnarray*}
Multiplying on the left and right by $I\otimes W^{T}$ and $I\otimes W$,
respectively, and applying the condition $W^{T}=W^{-1}$, we have
\begin{eqnarray*}
0 & = & \left[I\otimes W^{T}\right]\left(V^{-T}\otimes W^{-T}\right)\Lambda\left(V^{T}\otimes W^{T}\right)P\left[I\otimes W\right]\\
 &  & +\left[I\otimes W^{T}\right]P\left(V\otimes W\right)\Lambda\left(V^{-1}\otimes W^{-1}\right)\left[I\otimes W\right]+\left[I\otimes W^{T}\right]Q\otimes I\left[I\otimes W\right]\\
 & = & \left(V^{-T}\otimes I\right)\Lambda\left(V^{T}\otimes I\right)\left[I\otimes W^{T}\right]P\left[I\otimes W\right]\\
 &  & +\left[I\otimes W^{T}\right]P\left[I\otimes W\right]\left(V\otimes I\right)\Lambda\left(V^{-1}\otimes I\right)+Q\otimes I.
\end{eqnarray*}
Let $\tilde{P}=\left[I\otimes W^{T}\right]P\left[I\otimes W\right]$
then
\[
\left(V^{-T}\otimes I\right)\Lambda\left(V^{T}\otimes I\right)\tilde{P}+\tilde{P}\left(V\otimes I\right)\Lambda\left(V^{-1}\otimes I\right)=-Q\otimes I,
\]
equivalently after row/column permutations then
\begin{eqnarray*}
D\left[F_{is}\right]^{T}D\left[\tilde{P}_{is}\right]+D\left[\tilde{P}_{is}\right]D\left[F_{is}\right] & = & -Q\otimes I,
\end{eqnarray*}
where $D\left[F_{is}\right]=\left[\begin{array}{ccc}
F_{1s} &  & 0\\
 & F_{2s}\\
0 &  & \ddots
\end{array}\right]$. From Prop.~\ref{Prop: Eigenvalues of A2}, the eigenvectors of $F_{2}$
are $v_{1i}\otimes w_{i}$ and $v_{2i}\otimes w_{i}$ with $Lw_{i}=\lambda_{i}w_{i}$.
Consequently, for $i\neq1$ with $\lambda_{i}\neq0$ then $F_{is}=\frac{1}{2}\left[\begin{array}{cc}
-1 & x_{i}\\
4\alpha-x_{i} & 0
\end{array}\right]$ and for $i=1$ with $\lambda_{1}=0$ then $F_{is}=\frac{1}{2}\left[\begin{array}{cc}
-1 & 0\\
-4\alpha(1-y) & -y
\end{array}\right]$. For $\lambda_{i}\neq0$ , then

\begin{eqnarray*}
F_{is}^{T}\tilde{P}_{is}+\tilde{P}_{is}F_{is} & = & -\left[\begin{array}{cc}
0 & 0\\
0 & \sigma^{2}
\end{array}\right]\\
\frac{1}{2}\left[\begin{array}{cc}
-1 & 4\alpha-x_{i}\\
x_{i} & 0
\end{array}\right]\left[\begin{array}{cc}
p_{11} & p_{12}\\
p_{12} & p_{22}
\end{array}\right]+\frac{1}{2}\left[\begin{array}{cc}
p_{11} & p_{12}\\
p_{12} & p_{22}
\end{array}\right]\left[\begin{array}{cc}
-1 & x_{i}\\
4\alpha-x_{i} & 0
\end{array}\right] & = & -\left[\begin{array}{cc}
0 & 0\\
0 & \sigma^{2}
\end{array}\right]
\end{eqnarray*}
so
\[
\tilde{P_{is}}=-\frac{\sigma^{2}}{x_{i}}\left[\begin{array}{cc}
4\alpha-x_{i} & 1\\
1 & \left(1-4\alpha x_{i}+x_{i}^{2}\right)/\left(4\alpha-x_{i}\right)
\end{array}\right].
\]
Similarly, for $i=1$ then

\[
\tilde{P}_{1s}=\frac{\sigma^{2}}{y\left(1+y\right)}\left[\begin{array}{cc}
16\alpha^{2}\left(y-1\right)^{2} & 4\alpha\left(y-1\right)\\
4\alpha\left(y-1\right) & \left(1+y\right)
\end{array}\right]
\]
with $\lim_{y\rightarrow0}\tilde{P}_{1s}=\left[\begin{array}{cc}
0 & 0\\
0 & \infty
\end{array}\right],$ and its associated eigenvalue is $\left\{ 0,\infty\right\} .$

The trace of $P$ without the modes associated with $\left\{ 0,\infty\right\} $
denoted as $\mbox{tr}_{*}P$ is
\red{
\begin{eqnarray*}
\mbox{tr}_{*}P & = & \mbox{tr}_{*}\left[I\otimes W\right]\tilde{P}\left[I\otimes W^{T}\right]=\mbox{tr}_{*}\left[I\otimes W^{T}\right]\left[I\otimes W\right]\tilde{P}\\
 & = & \mbox{tr}_{*}\tilde{P}=\sum_{i=2}^{N}\mbox{tr}(\tilde{P}_{is})=-\sum_{i=2}^{N}\frac{\sigma^{2}}{x_{i}}\left(4\alpha-x_{i}+\frac{1-4\alpha x_{i}+x_{i}^{2}}{4\alpha-x_{i}}\right)\\
 & = & -\sum_{i=2}^{N}\frac{\sigma^{2}}{x_{i}}\left(4\alpha-x_{i}-x_{i}+\frac{1}{4\alpha-x_{i}}\right)=2\sigma^{2}\sum_{i=2}^{N}\left(1-\frac{2\alpha}{x_{i}}-\frac{1}{2x_{i}\left(4\alpha-x_{i}\right)}\right)\\
 & = & 2\sigma^{2}\left(N-1-\sum_{i=2}^{N}\left(\frac{2\alpha}{x_{i}}+\frac{1}{2x_{i}\left(4\alpha-x_{i}\right)}\right)\right).
\end{eqnarray*}
On the ring network, $x_{i}=\beta\cos\frac{2\pi k}{N}\lambda_{i}$
where $\left\{ \lambda_{2},\dots,\lambda_{N}\right\} =\left\{ 4\sin^{2}\frac{\pi}{N},4\sin^{2}\frac{2\pi}{N},\dots,4\sin^{2}\frac{\pi\left(N-1\right)}{N}\right\} $.
Using the relation $\sum_{i=1}^{N-1}\csc^{2}\frac{i\pi}{N}=(N^{2}-1)/3$,
this trace is further simplified as,
\begin{eqnarray*}
\mbox{tr}_{*}P & = & 2\sigma^{2}\left(N-1-\frac{\alpha(N^{2}-1)}{6\beta\cos\frac{2\pi k}{N}}-\frac{1}{32\beta\cos\frac{2\pi k}{N}}\sum_{i=1}^{N-1}\csc^{2}\frac{i\pi}{N}\frac{1}{\alpha-\beta\cos\frac{2\pi k}{N}\sin^{2}\frac{i\pi}{N}}\right)\\
 & = & 2\sigma^{2}\left(N-1\right)\left(1-\frac{\alpha(N+1)}{6\beta\cos\frac{2\pi k}{N}}-\frac{\alpha\left(N+1\right)}{6\beta\cos\frac{2\pi k}{N}}\Gamma(N,k,\alpha,\beta)\right),\\
 & = & 2\sigma^{2}\left(N-1\right)\left(1-\frac{\alpha\left(N+1\right)}{6\beta\cos\frac{2\pi k}{N}}\left[1+\Gamma(N,k,\alpha,\beta)\right]\right)
\end{eqnarray*}
where $\Gamma(N,k,\alpha,\beta)=\frac{3}{16\alpha(N^{2}-1)}\sum_{i=1}^{N-1}\csc^{2}\frac{i\pi}{N}\left(\alpha-\beta\cos\frac{2\pi k}{N}\sin^{2}\frac{i\pi}{N}\right)^{-1}$.

As $\beta\cos\frac{2\pi k}{N}\in\left[-1,0\right]$, $\left(\alpha-\beta\cos\frac{2\pi k}{N}\sin^{2}\frac{i\pi}{N}\right)^{-1}\in\left[\left(\alpha-\beta\cos\frac{2\pi k}{N}\right)^{-1},\alpha^{-1}\right]$
and $\sum_{i=1}^{N-1}\csc^{2}\frac{i\pi}{N}=(N^{2}-1)/3,$ then $\Gamma(N,k,\alpha,\beta)\in\left[\left(16\alpha\left(\alpha-\beta\cos\frac{2\pi k}{N}\right)\right)^{-1},\left(16\alpha^{2}\right)^{-1}\right]\subseteq\left(16\alpha\right)^{-1}\left[\left(\alpha+\left|\beta\right|\right)^{-1},\left(\alpha\right)^{-1}\right]$.
}

\red{
\section{Interval Exit Probability}
\label{App:ProbabilityInterval}

The perturbed edge states on a ring about an equilibrium defined by
phase offsets $\Delta_{k}$ is described by the states
\[
\delta e_{i}=\phi_{i+1}-\phi_{i}-\Delta_{k}=\delta\phi_{i+1}-\delta\phi_{i},
\]
or compactly by $\delta e=\left[\begin{array}{cc}
0 & E-I\end{array}\right]\left[\begin{array}{c}
\delta a\\
\delta\phi
\end{array}\right]$, where $E$ is defined in Appendix \ref{App:Covariance}. Consequently, the covariance
matrix $P_{e}=\mathbb{E}\left(\delta e\delta e^{T}\right)$ can be
found by a projection of the covariance matrix $P'=\mathbb{E}\left(\left[\begin{array}{c}
\delta a\\
\delta\phi
\end{array}\right]\left[\begin{array}{c}
\delta a\\
\delta\phi
\end{array}\right]^{T}\right)$ as
\[
P_{e}=\left[\begin{array}{cc}
0 & E-I\end{array}\right]P'\left[\begin{array}{cc}
0 & E-I\end{array}\right]^{T}.
\]
The trace of $P_{e}$ without the mode associated with
the undamped subspace spanned by $\delta e=\boldsymbol{1}$ is denoted
as $\mbox{tr}_{*}P_{e}$. From Appendix \ref{App:Covariance}, noting that $\left(E^{T}-I\right)\left(E-I\right)=L$,  
$\mbox{tr}_{*}(\left[\begin{array}{cc}
0 & E-I\end{array}\right]P'\left[\begin{array}{cc}
0 & E-I\end{array}\right]^{T})=\mbox{tr}_{*}(\left[\begin{array}{cc}
0 & E-I\end{array}\right]P\left[\begin{array}{cc}
0 & E-I\end{array}\right]^{T})$ and applying the closed form solution for $P$
then

\begin{eqnarray*}
\mbox{tr}_{*}P_{e} & = & \mbox{tr}_{*}\left(\left[\begin{array}{cc}
0 & E-I\end{array}\right]P\left[\begin{array}{c}
0\\
E^{T}-I
\end{array}\right]\right)=\mbox{tr}_{*}\left(\left[\begin{array}{cc}
0 & 0\\
0 & \left(E^{T}-I\right)\left(E-I\right)
\end{array}\right]\left(I\otimes W\right)\tilde{P}\left(I\otimes W^{T}\right)\right)\\
 & = & \mbox{tr}_{*}\left(\left[\begin{array}{cc}
0 & 0\\
0 & L
\end{array}\right]\left(I\otimes W\right)\tilde{P}\left(I\otimes W^{T}\right)\right)=\mbox{tr}_{*}\left(\left(I\otimes W^{T}\right)\left(\left[\begin{array}{cc}
0 & 0\\
0 & 1
\end{array}\right]\otimes L\right)\left(I\otimes W\right)\tilde{P}\right)\\
 & = & \mbox{tr}_{*}\left(\left(\left[\begin{array}{cc}
0 & 0\\
0 & 1
\end{array}\right]\otimes\Lambda\right)\tilde{P}\right)=\sum_{i=2}^{N}-\frac{\sigma^{2}}{x_{i}}\mbox{tr}_{*}\left(\left[\begin{array}{cc}
0 & 0\\
0 & \lambda_{i}
\end{array}\right]\left[\begin{array}{cc}
4\alpha-x_{i} & 1\\
1 & \left(1-4\alpha x_{i}+x_{i}^{2}\right)/\left(4\alpha-x_{i}\right)
\end{array}\right]\right)\\
 & = & \sum_{i=2}^{N}-\frac{\sigma^{2}}{x_{i}}\mbox{tr}\left(\left[\begin{array}{cc}
0 & 0\\
\lambda_{i} & \lambda_{i}\left(1-4\alpha x_{i}+x_{i}^{2}\right)/\left(4\alpha-x_{i}\right)
\end{array}\right]\right)=\sum_{i=2}^{N}-\frac{\sigma^{2}}{x_{i}}\frac{\lambda_{i}\left(1-4\alpha x_{i}+x_{i}^{2}\right)}{4\alpha-x_{i}}\\
 & = & -\frac{\sigma^{2}}{\beta c}\sum_{i=2}^{N}\frac{1-4\alpha x_{i}+x_{i}^{2}}{4\alpha-x_{i}}=\frac{\sigma^{2}}{\beta c}\sum_{i=2}^{N}x_{i}-\frac{1}{4\alpha-x_{i}}\\
 & = & \sigma^{2}\left(\sum_{i=2}^{N}\lambda_{i}-\frac{1}{\beta c}\sum_{i=2}^{N}\frac{1}{4\alpha-\lambda_{i}\beta c}\right).
\end{eqnarray*}
For the ring graph due to the underlying symmetry in the $\delta e_{i}$
states then $\bar{\sigma}^{2}:=\mathbb{E}(e_{1}^{2})=\mathbb{E}(e_{2}^{2})=\dots=\mathbb{E}(e_{N}^{2})$
and so $\bar{\sigma}^{2}=\mbox{tr}_{*}P_{e}/N$. For a ring graph
then $\left\{ \lambda_{2},\dots,\lambda_{N}\right\} = \\ \left\{ 4\sin^{2}\frac{\pi}{N},4\sin^{2}\frac{2\pi}{N},\dots,4\sin^{2}\frac{\pi\left(N-1\right)}{N}\right\} $ and
$\sum_{i=2}^{N}\lambda_{i}=2N$, so
\[
\bar{\sigma}^{2}=\sigma^{2}\left(2-\frac{1}{4\beta N\cos\frac{2\pi k}{N}}\sum_{i=1}^{N-1}\frac{1}{\alpha-\beta\cos\frac{2\pi k}{N}\sin^{2}\frac{i\pi}{N}}\right).
\]
Let the probability that the random variable $\delta e_{i}\sim \mathcal{N}(0,\bar{\sigma}^{2})$
remains in the bounded interval $\left[\varepsilon_{l},\varepsilon_{u}\right]$
be $p_{i}(\varepsilon_{l},\varepsilon_{u})$. This probability can
be calculated using the cumulative distribution function $F(\cdot)$ of
the Gaussian distribution and the error function $\mbox{erf}(\cdot)$
as
\begin{eqnarray*}
p_{i}(\varepsilon_{l},\varepsilon_{u}) & = & F(\varepsilon_{u})-F\left(\varepsilon_{l}\right)\\
 & = & \frac{1}{2}\left[1+\mbox{erf}\left(\frac{\varepsilon_{u}}{\bar{\sigma}\sqrt{2}}\right)\right]-\frac{1}{2}\left[1+\mbox{erf}\left(\frac{\varepsilon_{l}}{\bar{\sigma}_{k}\sqrt{2}}\right)\right]\\
 & = & \frac{1}{2}\left[\mbox{erf}\left(\frac{\varepsilon_{u}}{\bar{\sigma}\sqrt{2}}\right)-\mbox{erf}\left(\frac{\varepsilon_{l}}{\bar{\sigma}\sqrt{2}}\right)\right].
\end{eqnarray*}
Assuming that cross-coupling between $\delta e_{i}$'s are small, the probability
of all edge states remaining bounded $p(\varepsilon_{l},\varepsilon_{u})$
can be approximated as
\[
p(\varepsilon_{l},\varepsilon_{u})\approx p_{i}(\varepsilon_{l},\varepsilon_{u})^{N}=\frac{1}{2^{N}}\left[\mbox{erf}\left(\frac{\varepsilon_{u}}{\bar{\sigma}\sqrt{2}}\right)-\mbox{erf}\left(\frac{\varepsilon_{l}}{\bar{\sigma}\sqrt{2}}\right)\right]^{N}.
\]
The probability of exiting the interval $\left[\varepsilon_{l},\varepsilon_{u}\right]$
by time $T$ given a sampling interval $\Delta t$ is then
\begin{eqnarray*}
p_{\left[0,T\right]}(\varepsilon_{l},\varepsilon_{u}) & = & \sum_{k=1}^{\left\lfloor T/\Delta t\right\rfloor }p(\varepsilon_{l},\varepsilon_{u})^{k-1}p(\varepsilon_{l},\varepsilon_{u})\\
 & = & 1-p(\varepsilon_{l},\varepsilon_{u})^{\left\lfloor T/\Delta t\right\rfloor },
\end{eqnarray*}
and consequently the probability of first exiting in the time span $\left[t_{1},t_{2}\right]$
is
\[
p_{\left[t_{1},t_{2}\right]}(\varepsilon_{l},\varepsilon_{u})=p_{\left[0,t_{2}\right]}(\varepsilon_{l},\varepsilon_{u})-p_{\left[0,t_{1}\right]}(\varepsilon_{l},\varepsilon_{u})=p(\varepsilon_{l},\varepsilon_{u})^{\left\lfloor t_{1}/\Delta t\right\rfloor }-p(\varepsilon_{l},\varepsilon_{u})^{\left\lfloor t_{2}/\Delta t\right\rfloor }.
\]
Noting that the cumulative distribution function for this event is therefore $F(T)=p_{\left[0,T\right]}(\varepsilon_{l},\varepsilon_{u})$
the expected switching time is
\begin{eqnarray*}
\mathbb{E}_{T}(\varepsilon_{l},\varepsilon_{u}) & = & \int_{0}^{\infty}t\frac{d}{dt}F(t)dt\\
 & = & \int_{0}^{\infty}t\frac{d}{dt}\left(1-p(\varepsilon_{l},\varepsilon_{u})^{t/\Delta t}\right)dt\\
 & = & -\frac{\Delta t}{\log(p(\varepsilon_{l},\varepsilon_{u}))}.
\end{eqnarray*}
}

\end{document}